\documentclass[manuscript,nonacm]{acmart}
\usepackage{textcomp}
\usepackage{amscd}
\usepackage{amsfonts}
\usepackage{listings}
\usepackage{float}
\usepackage{hyperref}
\usepackage[vlined,linesnumbered,noresetcount]{algorithm2e}
\usepackage[binary-units]{siunitx}

\usepackage[inline]{enumitem}
\usepackage{tikz}
\usetikzlibrary{arrows.meta,patterns}
\usetikzlibrary{ipe}
\usepackage{tikzscale}
\usepackage{wrapfig}

\newtheorem{assumption}{Assumption}
\definecolor{darkgray}{gray}{0.30}
\definecolor{light-gray}{gray}{0.87}

\newcommand{\citeAppendix}[2]{#1}

\newcommand{\ie}{\emph{i.e.}\xspace}

\newcommand{\eg}{\emph{e.g.}\xspace}

\newcommand{\cf}{\textit{cf.}\xspace}

\newcommand{\cmp}{\emph{cmp.}\xspace}

\newcommand*{\blockdag}{block DAG\xspace}
\newcommand*{\blockdags}{block DAGs\xspace}

\newcommand*{\name}[1]{\textsc{#1}}

\newcommand{\blockmania}{\name{Blockmania}\xspace}
\newcommand{\hashgraph}{\name{Hashgraph}\xspace}
\newcommand{\stellar}{\name{Stellar}\xspace}
\renewcommand{\aleph}{\name{Aleph}\xspace}
\newcommand{\flare}{\name{Flare}\xspace}
\newcommand{\coq}{\name{Coq}\xspace}
\newcommand{\peerreview}{\name{PeerReview}\xspace}
\newcommand{\polygraph}{\name{Polygraph}\xspace}
\newcommand{\cassandra}{\name{Appache Cassandra}\xspace}
\newcommand{\awsse}{\name{Aws S3}\xspace}
\newcommand{\arxiv}{\name{arxiv}\xspace}

\newcommand*{\hsh}{\#}
\newcommand*{\valid}{\mathsf{valid}}

\newcommand*{\sle}[1]{\mathsf{#1}} % block element
\newcommand*{\inm}{\sle{in}}
\newcommand*{\outm}{\sle{out}}
\newcommand*{\buffer}{\mathsf{blks}}
\newcommand*{\timer}[1]{\Delta_{#1}}

\newcommand*{\requestBuffer}{\mathsf{rqsts}}
\newcommand*{\insertRequest}{\mathsf{put}}
\newcommand*{\takeRequests}{\mathsf{get}}

\newcommand*{\signs}{\mathsf{sign}}

\newcommand*{\authenticate}{\mathsf{authenticate}}

\newcommand*{\ble}[1]{\mathsf{#1}} % block element
\newcommand*{\nid}{\ble{n}}
\newcommand*{\seqn}{\ble{k}}
\newcommand*{\parent}{\ble{parent}}
\newcommand*{\preds}{\ble{preds}}

\newcommand*{\rs}{\ble{rs}}
\newcommand*{\concat}{\cdot}
\newcommand*{\signature}{\sigma}
\newcommand*{\signatures}{\Sigma}
\newcommand*{\empsign}{\ble{null}}
\newcommand*{\empseq}{[\ ]}
\newcommand*{\validsign}{\mathsf{verify}}

\newcommand*{\refB}{\mathsf{ref}}
\newcommand*{\bB}{\mathcal{B}} % buildee block

\newcommand*{\disseminate}{\mathsf{disseminate}}

\newcommand*{\srvrs}{\mathsf{Srvrs}}
\newcommand*{\byz}[1]{\check{#1}}

\newcommand*{\rqsts}{\mathsf{Rqsts}}
\newcommand*{\inds}{\mathsf{Inds}}
\newcommand*{\vals}{\mathsf{Vals}}

\newcommand*{\bool}{\mathbb{B}}
\newcommand*{\true}{\mathsf{true}}
\newcommand*{\false}{\mathsf{false}}
\newcommand*{\blks}{\mathsf{Blks}}

\newcommand*{\msgs}{\mathsf{M}}
\newcommand*{\emp}{\varnothing}
\newcommand*{\pow}[1]{2^{#1}}

\newcommand*{\blkdags}{\mathsf{Dags}}
\newcommand*{\ins}{\mathsf{insert}}

\makeatletter
\newcommand{\myuset}[3][0ex]{%
  \mathrel{\mathop{#3}\limits^{
    \vbox to#1{\kern+4.5\ex@
    \hbox{$\scriptscriptstyle#2$}\vss}}}}
\makeatother

\newcommand*{\vtcs}{\mathsf{V}}
\newcommand*{\edgs}{\mathsf{E}}
\newcommand*{\empG}{\varnothing}
\newcommand*{\nxt}{\rightharpoonup{}}

\newcommand*{\nxtN}[1]{\rightharpoonup^{#1}{}}
\newcommand*{\nxtT}{\nxtN{+}}
\newcommand*{\nxtS}{\nxtN{*}}

\newcommand*{\pis}{\mathsf{PIs}}
\newcommand*{\bs}{\textit{bs}}

\newcommand*{\ghost}[1]{\textcolor{black}{#1}}

\newcommand*{\lbls}{\mathcal{L}}

\newcommand*{\gossip}{\mathsf{gossip}}
\newcommand*{\interpret}{\mathsf{interpret}}

\newcommand*{\assign}{\mathrel{\coloneqq}}

\SetKwBlock{SubAlgoBlock}{}{end}
\newcommand{\SubAlgo}[2]{#1 \SubAlgoBlock{#2}}

\SetKw{Module}{module}
\SetKw{Function}{function}
\SetKw{Send}{send}
\SetKw{Broadcast}{broadcast}
\SetKw{If}{if}
\SetKw{Else}{else}
\SetKw{Upon}{when}
\SetKw{Received}{received}
\SetKw{Triggered}{triggered}
\SetKw{Trigger}{trigger}
\SetKw{From}{from}
\SetKw{Every}{every}
\SetKw{And}{and}
\SetKw{Not}{not}
\SetKw{Or}{or}
\SetKw{To}{to}
\SetKw{SuchThat}{such that}
\SetKw{Exists}{exists}
\SetKw{Exist}{exist}
\SetKw{Forall}{for all}
\SetKw{While}{while}
\SetKw{Pick}{pick}
\SetKw{Do}{do}
\SetKw{For}{for}
\SetKw{Some}{some}
\SetKw{Then}{then}
\SetKw{Where}{where}
\SetKw{Holds}{holds}
\SetKw{New}{new}
\SetKw{Timer}{timer}
\SetKw{Set}{set}
\SetKw{Start}{start}
\SetKw{After}{after}
\SetKw{Expires}{expires}
\SetKw{Time}{time}
\SetKw{True}{true}
\SetKw{Return}{return}
\SetKw{Process}{process}
\SetKw{Channel}{channel}
\SetKw{Copy}{copy}

\renewcommand{\G}{\mathcal{G}}
\newcommand*{\FWD}{\mathtt{FWD}}

\renewcommand{\P}{\mathcal{P}}

\newcommand*{\receive}{\mathsf{receive}}

\newcommand*{\sender}{\mathsf{sender}}
\newcommand*{\receiver}{\mathsf{receiver}}

\newcommand*{\brdcst}{\mathsf{broadcast}}
\newcommand*{\deliver}{\mathsf{deliver}}
\newcommand*{\echoed}{\mathsf{echoed}}
\newcommand*{\readied}{\mathsf{readied}}
\newcommand*{\delivered}{\mathsf{delivered}}
\newcommand*{\broadcast}{\mathsf{broadcast}}
\newcommand*{\ECHO}{\mathtt{ECHO}}
\newcommand*{\READY}{\mathtt{READY}}
\newcommand*{\In}{\mathsf{in}}
\newcommand*{\Out}{\mathsf{out}}

\newcommand*{\interpreted}{\mathcal{I}}
\newcommand*{\bfr}{\mathsf{Ms}}
\newcommand*{\eligible}{\mathsf{eligible}}

\newcommand*{\shim}{\mathsf{shim}}
\newcommand*{\indicate}{\mathsf{indicate}}
\newcommand*{\request}{\mathsf{request}}

\newcommand*{\IF}{\mathbb{I}}
\newcommand*{\PR}{\mathbb{P}}

\newcommand*{\gssp}{\textit{gssp}}
\newcommand*{\intprt}{\textit{intprt}}

\copyrightyear{2021}
\acmYear{2021}
\setcopyright{acmlicensed}\acmConference[PODC '21]{Proceedings of the 2021 ACM Symposium on Principles of Distributed Computing}{July 26--30, 2021}{Virtual Event, Italy}
\acmBooktitle{Proceedings of the 2021 ACM Symposium on Principles of Distributed Computing (PODC '21), July 26--30, 2021, Virtual Event, Italy}
\acmPrice{15.00}
\acmDOI{10.1145/3465084.3467930}
\acmISBN{978-1-4503-8548-0/21/07}
\begin{document}
\fancyhead{}

\title{Embedding a Deterministic BFT Protocol in a Block DAG}

\author{Maria A Schett}
\email{mail@maria-a-schett.net}
\affiliation{
  \institution{University College London}
  \city{London}
  \country{United Kingdom}
}

\author{George Danezis}
\email{g.danezis@ucl.ac.uk}
\affiliation{
  \institution{University College London}
  \city{London}
  \country{United Kingdom}
}

\begin{abstract}
  This work formalizes the structure and protocols underlying recent
distributed systems leveraging block DAGs, which are essentially
encoding Lamport's happened-before relations between blocks, as their
core network primitives. We then present an embedding of any
deterministic Byzantine fault tolerant protocol $\mathcal{P}$ to
employ a block DAG for interpreting interactions between servers. Our
main theorem proves that this embedding maintains all safety and
liveness properties of $\mathcal{P}$. Technically, our theorem is
based on the insight that a block DAG merely acts as an efficient
reliable point-to-point channel between instances of $\mathcal{P}$
while also using $\mathcal{P}$ for efficient message compression.

%%% Local Variables:
%%% mode: latex
%%% TeX-master: "paper"
%%% End:

\end{abstract}

% keywords, ACM classification and conference information can be omitted for submission

\maketitle

% \mas{@GD: proposed conventions:
%   \begin{itemize}
%   \item American English (ACM, letter format)
%   \item servers (not: nodes, processes)
%   \item \blockdag/\blockdags (macros, because couldn't decide how to
%     write it: block dag, block DAG, blockdag \ldots)
%   \item every server \emph{builds} a block, \ie $\bB$ in $\gossip$
%     (not creates, proposes)
%   \item every server eventually \emph{disseminates} a block, or sends
%     $\bB$ in $\gossip$ to everyone.
%   \end{itemize}
% }

\section{Introduction}

Recent interest in blockchain and cryptocurrencies has resulted in a
renewed interest in Byzantine fault tolerant consensus for state
machine replication, as well as Byzantine consistent and reliable
broadcast that is sufficient to build payment
systems~\cite{2019_guerraoui_et_al,2020_baudet_et_al}.
These systems have high demands on throughput. To meet this demand a
number of designs~\cite{2020_wang_et_al} depart from the traditional
setting and generalize the idea of a blockchain to a more generic
directed acyclic graph embodying Lamport's happened-before
relations~\cite{1978_lamport} between blocks. We refer to this
structure as a \emph{block DAG}. Instead of directly sending protocol
messages to each other, participants rely on the common higher level
abstraction of the \blockdag.
Moving from a chain to a graph structure allows for parallelization:
every participant can propose a block with transactions and not merely
a leader. As with blockchains, these blocks of transactions link
cryptographically to past blocks establishing an order between them.

Examples of such designs are \hashgraph~\cite{2016_baird} used by the Hedera network,
as well as \aleph~\cite{2019_gagol_et_al}, \blockmania~\cite{2018_danezis_et_al}, and
\flare~\cite{2019_rowan_et_al}. These works argue a number of advantages for
the \blockdag approach. First, maintaining a joint \blockdag is simple
and scalable, and can leverage widely-available distributed key-value
stores. Second, they report impressive performance results compared with traditional
protocols that materialize point-to-point messages as direct network messages. This
results from batching many transactions in each block; using a low number of
cryptographic signatures, having minimal overhead when running deterministic parts of
the protocol; using a common \blockdag logic
while performing network IO, and only applying the higher-level protocol logic off-line
possibly later; and as a result supporting running many instances of protocols in
parallel `for free'.

However, while the protocols may be simple and
performant when implemented, their specification, and arguments for
correctness, safety and liveness are far from simple.
Their proofs and arguments are usually inherently tied to their
specific applications and requirements, but both specification and
formal arguments of \hashgraph, \aleph, \blockmania, and \flare are
structured around two phases:
\begin{enumerate*}[label=\emph{(\roman*)}]
\item building a \blockdag, and
\item running a protocol on top of the \blockdag.
\end{enumerate*}
We generalize their arguments by giving an abstraction of a \blockdag
as a \emph{reliable point-to-point link}. We can then rely on this
abstraction to simulate a protocol~$\P$---as a black-box---on top of
this point-to-point link \emph{maintaining the safety and liveness
  properties} of $\P$.
We believe that this modular formulation of the underlying mechanisms
through a clear separation of the high-level protocol~$\P$ and the
underlying \blockdag allows for easy re-usability and strengthens the
foundations and persuasiveness of systems based on \blockdags.

In this work we present a formalization of a \blockdag, the protocols
to maintain a joint \blockdag, and its properties. We show that any
deterministic Byzantine fault tolerant~(BFT) protocol, can be embedded
in this \blockdag, while maintaining its safety and liveness
properties. We demonstrate that the advantageous properties of
\blockdag based protocols claimed by \hashgraph, \aleph, \blockmania,
and \flare ---such as the efficient message compression, batching of
signatures, the ability to run multiple instances `for free', and
off-line interpretation of the \blockdag---emerge from the generic
composition we present.  Therefore, the proposed composition not only
allows for straight forward correctness arguments, but also preserves
the claimed advantages of using a \blockdag approach, making it a
useful abstraction not only to analyze but also implement systems that
offer both high assurance and high performance.

%\begin{wrapfigure}{r}{0.45\linewidth}
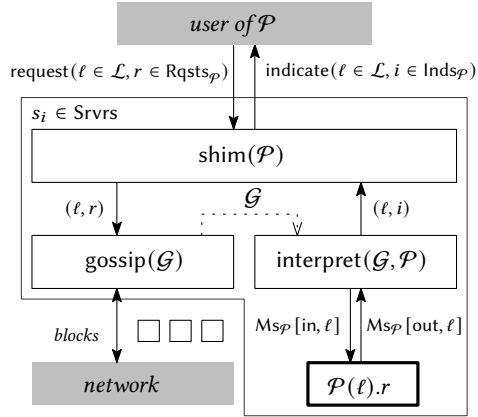
\begin{figure}[t]
  \centering
  \tikzstyle{ipe stylesheet} = [
  ipe import,
  even odd rule,
  line join=round,
  line cap=butt,
  ipe pen normal/.style={line width=0.4},
  ipe pen fat/.style={line width=1.2},
  ipe pen heavier/.style={line width=0.8},
  ipe pen ultrafat/.style={line width=2},
  ipe pen normal,
  ipe mark normal/.style={ipe mark scale=3},
  ipe mark large/.style={ipe mark scale=5},
  ipe mark small/.style={ipe mark scale=2},
  ipe mark tiny/.style={ipe mark scale=1.1},
  ipe mark normal,
  /pgf/arrow keys/.cd,
  ipe arrow normal/.style={scale=7},
  ipe arrow large/.style={scale=10},
  ipe arrow small/.style={scale=5},
  ipe arrow tiny/.style={scale=3},
  ipe arrow normal,
  /tikz/.cd,
  ipe arrows, % update arrows
  <->/.tip = ipe normal,
  ipe dash normal/.style={dash pattern=},
  ipe dash dotted/.style={dash pattern=on 1bp off 3bp},
  ipe dash dash dot dotted/.style={dash pattern=on 4bp off 2bp on 1bp off 2bp on 1bp off 2bp},
  ipe dash dash dotted/.style={dash pattern=on 4bp off 2bp on 1bp off 2bp},
  ipe dash dashed/.style={dash pattern=on 4bp off 4bp},
  ipe dash normal,
  ipe node/.append style={font=\normalsize},
  ipe stretch normal/.style={ipe node stretch=1},
  ipe stretch normal,
  ipe opacity 10/.style={opacity=0.1},
  ipe opacity 30/.style={opacity=0.3},
  ipe opacity 50/.style={opacity=0.5},
  ipe opacity 75/.style={opacity=0.75},
  ipe opacity opaque/.style={opacity=1},
  ipe opacity opaque,
]
\definecolor{red}{rgb}{1,0,0}
\definecolor{blue}{rgb}{0,0,1}
\definecolor{brown}{rgb}{0.647,0.165,0.165}
\definecolor{darkblue}{rgb}{0,0,0.545}
\definecolor{darkcyan}{rgb}{0,0.545,0.545}
\definecolor{darkgray}{rgb}{0.663,0.663,0.663}
\definecolor{darkgreen}{rgb}{0,0.392,0}
\definecolor{darkmagenta}{rgb}{0.545,0,0.545}
\definecolor{darkorange}{rgb}{1,0.549,0}
\definecolor{darkred}{rgb}{0.545,0,0}
\definecolor{gold}{rgb}{1,0.843,0}
\definecolor{gray}{rgb}{0.745,0.745,0.745}
\definecolor{green}{rgb}{0,1,0}
\definecolor{lightblue}{rgb}{0.678,0.847,0.902}
\definecolor{lightcyan}{rgb}{0.878,1,1}
\definecolor{lightgray}{rgb}{0.827,0.827,0.827}
\definecolor{lightgreen}{rgb}{0.565,0.933,0.565}
\definecolor{lightyellow}{rgb}{1,1,0.878}
\definecolor{navy}{rgb}{0,0,0.502}
\definecolor{orange}{rgb}{1,0.647,0}
\definecolor{pink}{rgb}{1,0.753,0.796}
\definecolor{purple}{rgb}{0.627,0.125,0.941}
\definecolor{seagreen}{rgb}{0.18,0.545,0.341}
\definecolor{turquoise}{rgb}{0.251,0.878,0.816}
\definecolor{violet}{rgb}{0.933,0.51,0.933}
\definecolor{yellow}{rgb}{1,1,0}
\definecolor{black}{rgb}{0,0,0}
\definecolor{white}{rgb}{1,1,1}
\begin{tikzpicture}[ipe stylesheet]
  \node[ipe node]
     at (190.717, 715.679) {$\shim(\P)$};
  \node[ipe node]
     at (148, 676) {$\gossip(\G)$};
  \node[ipe node]
     at (220, 676) {$\interpret(\G, \P)$};
  \draw[shift={(128, 728)}, yscale=1.25]
    (0, 0) rectangle (160, -16);
  \draw[shift={(128, 688)}, xscale=0.475, yscale=1.25]
    (0, 0) rectangle (160, -16);
  \draw[shift={(212, 688)}, xscale=0.475, yscale=1.25]
    (0, 0) rectangle (160, -16);
  \draw[shift={(232.001, 640)}, xscale=1.6884, yscale=0.8, ipe pen fat]
    (0, 0) rectangle (24.545, -20);
  \node[ipe node]
     at (239.206, 628.324) {$\P(\ell).r$};
  \draw[{ipe ptarc[ipe arrow small]}-]
    (160, 688)
     -- (160, 708);
  \draw[{ipe ptarc[ipe arrow small]}-]
    (252, 708)
     -- (252, 688);
  \draw[{ipe ptarc[ipe arrow small]}-]
    (248, 640)
     -- (248, 668);
  \draw[{ipe ptarc[ipe arrow small]}-]
    (252, 668)
     -- (252, 640);
  \draw[darkgray]
    (124, 740)
     -- (124, 664)
     -- (208, 664)
     -- (208, 620)
     -- (292, 620)
     -- (292, 740)
     -- (124, 740)
     -- (124, 740);
  \fill[lightgray]
    (128, 640) rectangle (204, 624);
  \node[ipe node]
     at (147.038, 628.962) {\emph{network}};
  \draw[{ipe ptarc[ipe arrow small]}-{ipe pointed[ipe arrow small]}]
    (160, 640)
     -- (160, 668);
  \node[ipe node, font=\footnotesize]
     at (136, 648) {\emph{blocks}};
  \draw[shift={(168, 656)}, scale=0.5]
    (0, 0) rectangle (16, -16);
  \draw[shift={(180, 656)}, scale=0.5]
    (0, 0) rectangle (16, -16);
  \draw[shift={(192, 656)}, scale=0.5]
    (0, 0) rectangle (16, -16);
  \node[ipe node, font=\footnotesize]
     at (253.915, 652) {$\bfr_\P[\Out, \ell]$};
  \node[ipe node, font=\footnotesize]
     at (212, 652) {$\bfr_\P[\In, \ell]$};
  \node[ipe node, font=\footnotesize]
     at (119.753, 747.583) {$\request(\ell \in \lbls, r \in \rqsts_\P)$};
  \node[ipe node, font=\footnotesize]
     at (216, 748) {$\indicate(\ell \in \lbls, i \in \inds_\P)$};
  \node[ipe node, font=\footnotesize]
     at (256, 696) {$(\ell, i)$};
  \node[ipe node, font=\footnotesize]
     at (140, 696) {$(\ell, r)$};
  \draw[{ipe ptarc[ipe arrow small]}-]
    (204, 728)
     -- (204, 760);
  \draw[{ipe ptarc[ipe arrow small]}-]
    (212, 760)
     -- (212, 728);
  \fill[shift={(160, 776)}, xscale=1.2632, lightgray]
    (0, 0) rectangle (76, -16);
  \node[ipe node]
     at (187.038, 764.962) {\emph{user of $\P$}};
  \draw[ipe dash dotted, -{ipe linear[ipe arrow small]}]
    (192, 688)
     -- (192, 696)
     -- (228, 696)
     -- (228, 688);
  \node[ipe node]
     at (208, 700) {$\G$};
  \node[ipe node, font=\small]
     at (128, 732) {$s_i \in \srvrs$};
\end{tikzpicture}
  \caption{Components and interfaces.}
  \label{fig:components}
\end{figure}
%\end{wrapfigure}

\paragraph{Overview} %
Figure~\ref{fig:components} shows the interfaces and components of our
proposed \blockdag framework parametric by a deterministic BFT
protocol~$\P$.
At the top, we have a \emph{user} seeking to run one or multiple instances
of~$\P$ on servers~$\srvrs$.
First, to distinguish between multiple \emph{protocol instances} the
user assigns them a \emph{label}~$\ell$ from a set of
\emph{labels}~$\lbls$.
Now, for $\P$ there is a set of possible
\emph{requests}~$\rqsts_\P$. But instead of requesting
$r \in \rqsts_\P$ from $s_i \in \srvrs$ running $\P$ for protocol
instance~$\ell$, the user calls the high-level interface of our
\blockdag framework: $\request(\ell, r)$ in $\shim(\P)$.
Internally, $s_i$ passes $(\ell, r)$ on to $\gossip(\G)$---which
continuously builds $s_i$'s \emph{\blockdag}~$\G$ by receiving and
disseminating \emph{blocks}. The passed $(\ell, r)$ is included into
the next block $s_i$ disseminates, and $s_i$ also includes references
to other received blocks, where cryptographic primitives prevent
byzantine servers from adding cycles between
blocks~\cite{2002_maniatis_et_al}.
These blocks are continuously exchanged by the servers utilizing the
low-level interface to the \emph{network} to exchange
blocks. In Section~\ref{sec:block-dags} we
formally define the \blockdag, its properties and protocols for
servers to maintain a joint \blockdag.
Independently, indicated by the dotted line, $s_i$ interprets $\P$ by
\emph{reading} $\G$ and running $\interpret(\G, \P)$.
To do so, $s_i$ locally simulates every protocol instance~$\P$ with
label~$\ell$ by simulating one process instance~of $\P(\ell)$ for
every server~$s \in \srvrs$.
To drive the simulation, $s_i$ passes the request $r$ read from a
block in $\G$ to $\P$, and then $s_i$ simulates the message exchange between any
two servers based on the structure of the \blockdag and the
deterministic protocol $\P$. Therefore $s_i$ moves messages between
in- and out-buffers $\bfr_\P[\In, \ell]$ and $\bfr_\P[\Out, \ell]$.
%
% \gd{NOT POSSIBLE TO UNDERSTAND AT THIS POINT, CUT OR MOVE TO
% INTERPRET: To do so, $s_i$ moves messages between process instances
% through the buffers $\bfr_\P[\In, \ell]$ and $\bfr_\P[\Out, \ell]$
% by interpreting the blocks and edges in $\G$: if $s_1$ sent a
% block~$B_1$ and $s_2$ sent a block~$B_2$ with a reference to $B_1$,
% then $s_i$ interprets this as a message sent from $s_1$ to $s_2$.
% %
% But how does $s_i$ know which message~$m$ was sent from $s_1$ to
% $s_2$---as only blocks are exchanged over the network?
% %
% Either the message was triggered by the request~$r$ for the process
% instance~$\ell$ included in $B$, or $m$ was triggered based on the
% messages $s_1$ received up to this point. Now it important, that $\P$
% is deterministic. As $s_i$ knows the starting point---the request
% $r$---and all subsequently exchanged messages through
% interpreting~$\G$, $s_i$ can compute all following messages.
% %
% In this fashion, $s_i$ can simulate running $\P$ for all $s_i$ and}
%
Eventually, the simulation $\P(\ell)$ of the server~$s_i$ will
indicate~$i$ from the set of possible \emph{indications}~$\inds_\P$.
We show how the \blockdag essentially acts as a reliable
point-to-point link and describe how any deterministic BFT protocol $\P$ can be
interpreted on a \blockdag in Section~\ref{sec:interpret}.
Finally, after $\interpret$ indicated $i$, $\shim(\P)$ can indicate
$i$ for~$\ell$ to the user of $\P$.
From the user's perspective, the embedding of $\P$ acted as $\P$, \ie
$\shim(\P)$ \emph{maintained $\P$'s interfaces and properties}.
We prove this in Section~\ref{sec:interfacing} and illustrate the
\blockdag framework for $\P$ instantiated with byzantine reliable
broadcast protocol.
%
% The remaining of this paper if organized as follows: in Section
% \ref{sec:background} we introduce crypotgraphic primitives, and our
% notation around graphs, and the basic security and network assumptions
% underlying BFT protocols.
%
%
We give related work in Section~\ref{sec:related-work}, and conclude
in Section~\ref{sec:conclusion}, where we discuss integration aspects
of higher-level protocols and the \blockdag framework---including
challenges in embedding protocols with non-determinism, more advanced
cryptography, and BFT protocols operating under partial synchrony.

This is the short version of the paper, where we omit proof details
and the appendix. Please find the full version on
\arxiv~\cite{2021_schett_et_al}.

\paragraph{Contributions} %
We show that using the \blockdag framework of
Figure~\ref{fig:components} for a deterministic BFT protocol~$\P$
maintains the
\begin{enumerate*}[label=\emph{(\roman*)}]
\item interfaces, and
\item safety and liveness properties
\end{enumerate*}
of $\P$ (Theorem~\ref{th:meat}).
The argument is generic: interpreting the eventually \emph{joint
  \blockdag} implements a \emph{reliable point-to-point link}
(Lemma~\ref{th:joint-dag}, Lemma~\ref{lem:point-to-point}).  Using
this reliable point-to-point link any server can locally run a
simulation of $\P$ as a black-box. This simulation is an execution of
$\P$ and thus retains the properties of $\P$.
By using the \blockdag framework, the user gains efficient message
compression and runs instances of $\P$ in parallel `for free’.
This is due to the determinism of $\P$, which allows every server to
locally deduce message contents from only the fact that a message has
been received without actually receiving the contents. So every
received block automatically compresses these messages for preceding
blocks.
Moreover, with every new block every server creates a new instance of
$\P$.
%

%%% Local Variables:
%%% mode: latex
%%% TeX-master: "paper"
%%% End:

\section{Background} \label{sec:background}

\paragraph{System Model}
%
% servers
%
We assume a finite set of \emph{servers}~$\srvrs$. A \emph{correct}
server~$s \in \srvrs$ faithfully follows a protocol~$\P$. When $s$ is
\emph{byzantine}, then $s$ behaves arbitrarily. However, we assume
byzantine servers are computationally bounded (\eg, $s$ cannot forge
signatures, or find collisions in cryptographic hash functions) and
cannot interfere with the Trusted Computing Base of correct servers
(\eg, kill the electricity supply of correct servers). The set
$\srvrs$ is fixed and known by every $s' \in \srvrs$ and we assume
$3f + 1$ servers to tolerate at most $f$ byzantine servers.
%
% messages
%
The set of all messages in protocol $\P$ is $\msgs_\P$.  Every
message~$m \in \msgs_\P$ has a $m.\sender$ and a $m.\receiver$. We
assume an arbitrary, but fixed, total order on messages: $<_\msgs$.
%
% determinstic
%
A protocol $\P$ is \emph{deterministic} if a state~$q$ and a sequence
of messages~$m \in \msgs_\P$ determine state~$q'$ and out-going
messages~$M \subseteq \pow{\msgs_\P}$. In particular, deterministic
protocols do not rely on random behavior such as coin-flips.
%
% network
%
The exact requirements on the network synchronicity depend on the
protocol~$\P$, that we want to embed, \eg, we may require partial
synchrony~\cite{1988_dwork_et_al} to avoid FLP
\cite{1985_fischer_et_al}. The only network assumption we impose for
building \blockdags is the following:
\begin{assumption}[Reliable Delivery] \label{ass:reliable-links} %
  For two correct servers~$s_1$ and $s_2$, if $s_1$ sends a
  block~$B$ to $s_2$, then eventually $s_2$ receives $B$.
\end{assumption}

\paragraph{Cryptographic Primitives}
We assume a \emph{secure cryptographic hash function}~$\hsh: A \to A'$
and write $\hsh(x)$ for the hash of $x \in A$, and $\hsh(A)$ for $A'$
\citeAppendix{(Definition~\ref{def:crypto-hash})}{Definition~A.1}.
We further assume a secure \emph{cryptographic
  signature} scheme~\cite{2007_katz_et_al}:
given a set of \emph{signatures} $\signatures$ we have functions
$\signs : \srvrs \times \msgs \to \signatures$ and
$\validsign: \srvrs \times \msgs \times \signatures \to \bool$, where
$\validsign(s, m, \sigma) = \true$ iff $\signs(s, m) = \sigma$. Given
computational bounds on all participants appropriate parameters for
these schemes can be chosen to make their probability of failure
negligible, and for the remainder of this work we assume their
probability of failure to be zero.

\paragraph{Directed Acyclic Graphs}

A \emph{directed graph}~$\G$ is a pair of \emph{vertices}~$\vtcs$ and
\emph{edges}~$\edgs \subseteq \vtcs \times \vtcs$. We write $\empG$
for the \emph{empty graph}.
% reachable
If there is an edge from $v$ to $v'$, that is $(v, v') \in \edgs$, we
write $v \nxt v'$. If $v'$ is \emph{reachable} from $v$, then $(v,v')$
is in the transitive closure of $\nxt$, and we write $\nxtT$. We write
$\nxtS$ for the reflexive and transitive closure, and $v \nxtN{n} v'$
for $n \geqslant 0$ if $v'$ is reachable from $v$ in $n$ steps.
A graph $\G$ is \emph{acyclic}, if $v \nxtT v'$ implies $v \neq v'$ for
all nodes $v, v' \in \G$.
We abbreviate $v \in \G$ if $v \in \vtcs_\G$, and $V \subseteq \G$ if
$v \in \G$ for all $v \in V$.
Let $\G_1$ and $\G_2$ be directed graphs. We define $\G_1 \cup \G_2$
as $(\vtcs_{\G_1} \cup \vtcs_{\G_2}, \edgs_{\G_1} \cup \edgs_{\G_2})$,
and $\G_1 \leqslant \G_2$ holds if
$\vtcs_{\G_1} \subseteq \vtcs_{\G_2}$ and
$\edgs_{\G_1} = {\edgs_{\G_2} \cap (\vtcs_{\G_1} \times
  \vtcs_{\G_1})}$.
Note, for $\leqslant$ we not only require
$\edgs_{\G_1} \subseteq \edgs_{\G_2}$, but additionally $\edgs_{\G_1}$
must already contain all edges from $\edgs_{\G_2}$ between vertices in
$\G_1$.
The following definition for inserting a new vertex~$v$ is restrictive: it
permits to extend $\G$ only by a vertex~$v$ and edges to this $v$.

\begin{definition} \label{def:insert} %
  Let $\G$ be a directed graph, $v$ be a vertex, and a $E$ be a set of
  edges of the form $\{(v_i, v) \mid v_i \in V \subseteq \G \}$.  We
  define $\ins(\G, v, E) = (\vtcs_\G \cup \{ v \}, \edgs_\G \cup E)$.
\end{definition}

This unconventional definition of inserting a vertex is sufficient for
building a \blockdag---and helps to establish useful properties of the
\blockdag in the next lemma: (\ref{lem:ins-idempotent})~inserting a vertex is
idempotent, (\ref{lem:g:subseteq})~the original graph is a
subgraph of the graph with a newly inserted
vertex, and~(\ref{lem:g:acyclic}) a \blockdag is acyclic by
construction.

\begin{lemma} \label{lem:ins:properties} %
  For a directed graph~$\G$, a vertex~$v$, and a set of
  edges~$E = \{(v_i, v) \mid v_i \in V \subseteq \G \}$, the following
  properties of $\ins(\G, v, E)$ hold:
  \begin{enumerate*}
  \item \label{lem:ins-idempotent} %
    if $v \in \G$ and $E \subseteq \edgs_\G$, then
    $\ins(\G, v, E) = \G$;
  \item \label{lem:g:subseteq} %
    if $E = \{(v_i, v) \mid v_i \in V \subseteq \G \}$ and
    $v \not\in \G$, then $\G \leqslant \ins(\G, v, E)$; and
  \item \label{lem:g:acyclic} %
    if $\G$ is acyclic, $v \not\in \G$, then $\ins(\G, v, E)$ is
    acyclic.
  \end{enumerate*}
\end{lemma}

To give some intuitions,
for~Lemma~\ref{lem:ins:properties}~(\ref{lem:g:subseteq}), if
$v \in \G$ and $\G' = \ins(\G, v, E)$, then
${\edgs_{\G'} \cap (\vtcs_{\G} \times \vtcs_{\G})} = \edgs_{\G}$ may
not hold. For example, let $\G$ have vertices $v_1$ and $v_2$ with
$\edgs_{\G} = \varnothing$, and
$\G' = \ins(\G, v_2, \{ (v_1, v_2) \})$ with
$\edgs_{\G'} = \{ (v_1, v_2) \}$. Then we have
$\edgs_{\G} \neq {\edgs_{\G'} \cap (\vtcs_{\G} \times \vtcs_{\G})}$.
For~Lemma~\ref{lem:ins:properties}~(\ref{lem:g:acyclic}), if
$v \in \G$, then $\ins(\G, v, E)$ may add a cycle. For example, take
$\G$ with vertices $\{ v_1, v_2 \}$ and
$\edgs_{\G} = \{ (v_1, v_2) \}$ then $\ins(\G, v_1, \{ (v_2, v_1) \})$
contains a cycle.

%%% Local Variables:
%%% mode: latex
%%% TeX-master: "paper"
%%% End:

\section{Building a Block DAG} \label{sec:block-dags}

The networking component of the \blockdag protocol is very simple: it
has one core message type, namely a \emph{block}, which is constantly
disseminated. A block contains authentication for references to
previous blocks, requests associated to instances of protocol $\P$,
meta-data and a signature. Servers only exchange and validate blocks.
From these blocks with their references to previous blocks, servers
build their \blockdags.
Although servers build their \blockdags locally, eventually correct
servers have a joint \blockdag~$\G$. As we show in the next
Section~\ref{sec:interpret}, the servers can then \emph{independently}
interpret $\G$ as multiple instances of $\P$.

%\subsection{Blocks}

\begin{definition} \label{def:block} %
  A \emph{block}~$B \in \blks$ has
  \begin{enumerate*}[label=\emph{(\roman*)}]
  \item an \emph{identifier}~$\nid$ of the server~$s$ which built~$B$,
  \item a \emph{sequence number}~$\seqn \in \mathbb{N}_0$,
  \item a finite list of hashes of \emph{predecessor}
    blocks~$\preds = [ \refB(B_1), \ldots, \refB(B_k) ]$,
  \item a finite list of labels and
    requests~$\rs \in \pow{\lbls \times \rqsts}$, and
  \item a \emph{signature} $\signature = \signs(\nid, \refB(B))$.
  % for extending the definition
  \newcounter{block-elements}
  \setcounter{block-elements}{\value{enumi}}
  \end{enumerate*}
  Here, $\refB$ is a secure cryptographic hash function computed from
  $\nid$, $\seqn$, $\preds$, and $\rs$, but not $\signature$. By not
  depending on $\signature$, $\signs(B.\nid, \refB(B))$ is well defined.
\end{definition}

We use $B$ and $\refB(B)$ interchangeably, which is justified by
\emph{collision resistance} of $\refB$ %
\citeAppendix{(Definition
  \ref{def:crypto-hash}(\ref{it:collision-resistance}))}{Definition~A.1(3)}.
We use register notation, \eg, $B.\nid$ or $B.\signature$, to refer
to elements of a block $B$, and abbreviate
$B' \in \{ B' \mid \refB(B') \in B.\preds \}$ with $B' \in B.\preds$.
Given blocks $B$ and $B'$ with $B.\nid = B'.\nid$ and
$B'.\seqn = {B.\seqn + 1}$. %
If $B \in B'.\preds$ then we call $B$ a \emph{parent} of $B'$ and
write $B'.\parent = B$. We require that every block has at most one
parent.
We call $B$ a \emph{genesis block} if $B.\seqn = 0$. A genesis
block~$B$ cannot have a parent block, because $B.\seqn = 0$ and $0$ is
minimal in $\mathbb{N}_0$.

\begin{lemma} \label{lem:mutual-excl-refs} %
  For blocks $B_1$ and $B_2$, if $B_1 \in B_2.\preds$ then
  $B_2 \not\in B_1.\preds$.
\end{lemma}
Lemma~\ref{lem:mutual-excl-refs} prevents a byzantine server $\byz{s}$
to include a cyclic reference between $\byz{B}$ and $B$ by
\begin{enumerate*}
\item waiting for---or building itself---a block $B$ with
  $\refB(\byz{B}) \in B.\preds$, and then
\item building a block $\byz{B}$ such that $\refB(\byz{B}) \in B$.
\end{enumerate*}
As with secure time-lines~\cite{2002_maniatis_et_al},
Lemma~\ref{lem:mutual-excl-refs} gives a temporal ordering on $B$ and
$\byz{B}$. This is a static, cryptographic property, based on the
security of hash functions, and not dependent on the order in which
blocks are received on a network. While this prevents byzantine
servers from introducing cycles, they can still build ``faulty''
blocks. %
So next we define three checks for a server to ascertain that a block
is well-formed. If a block passes these checks, the block is
\emph{valid} from this server's point of view and the server
\emph{validated} the block.

\begin{definition} \label{def:valid} %
  A server~$s$ considers a block~$B$ \emph{valid}, written
  $\valid(s, B)$, if
  \begin{enumerate*}[label=\emph{(\roman*)}]
  \item \label{valid:sign-verified} %
    $s$ confirms $\validsign(B.\nid, B.\signature)$, \ie, that $B.\nid$
    built $B$,
  \item\label{valid:seqn} either %
    \begin{enumerate*}[label=\emph{(\alph*)}]
    \item \label{valid:genesis} %
      $B$ is a genesis block, or
    \item \label{valid:parent} %
      $B$ has exactly one parent, and
    \end{enumerate*}
  \item \label{valid:preds} %
    $s$ considers all blocks
    $B' \in  B.\preds$ valid.
  \end{enumerate*}
\end{definition}

\begin{figure}[t]
  \begin{minipage}[t]{0.45\linewidth}
    \centering
    \input{pics/block-dag}
    \caption{A \blockdag with 3 blocks $B_1$, $B_2$, and $B_3$.}
    \label{fig:bg}
  \end{minipage}
  \hspace{0.5cm}
  \begin{minipage}[t]{0.45\linewidth}
    \centering
    \input{pics/block-dag-lying}
    \caption{A \blockdag, where $\byz{s_1}$ is equivocating on the
      blocks $B_3$ and $B_4$.}
    \label{fig:bg-ly}
  \end{minipage}
\end{figure}

Especially, \ref{valid:seqn} deserves our attention: a server
$\byz{s_1}$ may still build two different blocks having the same
parent. However, $\byz{s_1}$ will not be able to create a further
block to `join' these two blocks with a different parent---their
successors will remain split. Essentially, this forces a linear
history from every block.

We assume, that if a correct server~$s$ considers a block~$B$ valid, then
$s$ can forward any block~$B' \in B.\preds$. That is, $s$ has received
the full content of $B'$---not only $\refB(B')$---and persistently
stores $B'$.
From valid blocks and their predecessors, a correct server builds a
\emph{\blockdag}:
\begin{definition} \label{def:block-dag} %
  For a server $s$, a \emph{\blockdag} $\G \in \blkdags$ is a directed
  acyclic graph with vertices $\vtcs_\G \subseteq \blks$, where
  \begin{enumerate*}[label=\emph{(\roman*)}]
  \item \label{bdag:all-valid} %
    $\valid(s, B)$ holds for all $B \in \vtcs_\G$, and
  \item \label{bdag:correct-edges} %
    if $B \in B'.\preds$ then $B \in \vtcs_\G$ and $(B, B') \in \edgs_\G$
    holds for all $B' \in \vtcs_\G$.
  \end{enumerate*}
  Let $B'$ be a block such that $\valid(s, B')$ holds and $B \in \G$
  for all $B \in B'.\preds$. Then $s$ \emph{inserts} $B'$ into $\G$ by
  $\ins(\G, B', \{ (B, B') \mid B \in B'.\preds \})$ after
  Definition~\ref{def:insert} and we write $\G.\ins(B)$. The
  preconditions guarantee that $\G.\ins(B')$ is a \blockdag
  \citeAppendix{(Lemma~\ref{lem:insert-preserves-block-dag})}{Lemma~A.3}.
\end{definition}

\begin{example}
  In Figure~\ref{fig:bg} we show a \blockdag with three blocks $B_1$,
  $B_2$, and $B_3$, where
  $B_1 = \{ \nid = s_1, \seqn = 0, \preds = \empseq \}$,
  $B_2 = \{ \nid = s_2, \seqn = 0, \preds = \empseq \}$, and
  $B_3 = \{ \nid = s_1, \seqn = 1, \preds = [ \refB(B_1), \refB(B_2)]
  \}$.
  Here, $\parent(B_3) = B_1$.
  %
  % We indicate the servers, which built the block in the rows, and the
  % sequence number of the blocks in the columns, and omit
  % requests~$\rs$ and signatures~$\sigma$ as they are not relevant in
  % this example
  %
  % The edges in the \blockdag point forward as opposed to the backwards
  % references to the predecessor blocks. An intution is that the
  % \blockdag grows by adding new blocks and edges point towards the
  % future. Pointers forward also fit with the reachability relation on
  % vertices in the graph: $B_1 \nxt B_3$.
  %
  Consider now Figure~\ref{fig:bg-ly} adding the block:
  $B_4 = \{ \nid = s_1, \seqn = 1, \preds = [ \refB(B_1), \refB(B_2) ]
  \}$.
  While all blocks in Figure~\ref{fig:bg-ly} are valid, with
  block~$B_4$, $\byz{s_1}$ is \emph{equivocating} on the block
  $B_3$---and vice versa.
\end{example}

%\subsection{Building Block DAGs}

\begin{algorithm}[t]
  \setcounter{AlgoLine}{0} %
  \SubAlgo{ %
    \Module $\gossip(s \in \srvrs, \G \in \blkdags, %
    \requestBuffer \in \pow{\lbls \times \rqsts})$ \label{lin:gossip} } %
  { $\bB \assign %
    \{ \nid : s, \seqn : 0, \preds : \empseq, \rs : \empseq,
    \signature : \empsign \} \in \blks$ \label{lin:gossip:bb} \\
    $\buffer \assign \emp \in \pow{\blks}$ \label{lin:gossip:buffer} \\
    \smallskip %
    \SubAlgo{ %
      \Upon \Received $B \in \blks$ \And $B\not\in
      \G$ \label{lin:gossip:rcv-b}%
    } %
    { %
      $\buffer \assign {\buffer \cup \{ B \} }$ \label{lin:gossip:ins-b-buffer} \\
    } %
    \smallskip %
    \SubAlgo { %
      \Upon $\valid(s, B')$ \For \Some $B' \in
      \buffer$ \label{lin:gossip:valid} %
    } %
    {
      $\G.\ins(B')$ \label{lin:gossip:ins-b-g} \\
      $\bB.\preds \assign {\bB.\preds \concat [ \refB(B') ]} $ \label{lin:gossip:ins-pred} \\
      $\buffer \assign {\buffer \setminus \{ B' \}}$ \label{lin:gossip:rm-b-buffer} \\
    } %
    \smallskip %
    \SubAlgo{ %
      \Upon $B' \in \buffer$ \And $B \in B'.\preds$ \Where
      $B \not\in \buffer$ \And $B \not\in
      \G$ \label{lin:gossip:missing-preds} \label{lin:gossip:deltaB} %
    } %
    { %
      \Send $\FWD\ \refB(B)$ \To $B'.\nid$ \label{lin:gossip:send-fwd} %
    }
    \smallskip %
    \SubAlgo{ %
      \Upon \Received $\FWD\ \refB(B)$ \From $s'$ \And $B \in
      \G$ \label{lin:gossip:rcv-fwd} } %
    { %
      \Send $B$ \To $s'$ \label{lin:gossip:fwd} %
    } %
    \smallskip %
    \SubAlgo{ %
      \Upon $\disseminate()$ \label{lin:gossip:disseminate}%
    } %
    { %
      $\bB \assign \{ \bB \textbf{ with } %
        \rs : \requestBuffer.\takeRequests(), %
        \signature : \signs(s, \bB)\}$  \label{lin:gossip:inject} \label{lin:gossip:sign} \\ %
      $\G.\ins(\bB)$ \label{lin:gossip:insB} \\ %
      \Send $\bB$ \To \Every $s' \in \srvrs$ \label{lin:gossip:brdcst}
      \\ %
      $\bB \assign \{ %
      \nid : s, \seqn : \bB.\seqn + 1, \preds : [ \refB(\bB) ], \rs : \empseq, $
       $ \signature : \empsign \} $ \label{lin:gossip:next-b}  \\
    } %
  }
  \caption{Building the \blockdag~$\G$ and block~$\bB$.}
  \label{alg:gossip}
\end{algorithm}

To build a \blockdag and blocks every correct server follows the
$\gossip$ protocol in Algorithm~\ref{alg:gossip}. By building a
\blockdag every correct server will eventually have a joint view on
the system. By building a block, every server can inject messages into
the system: either explicit messages from the high-level protocol by
directly writing those into the block, or implicit messages by adding
references to other blocks.
In Algorithm~\ref{alg:gossip}, a server~$s$ builds
\begin{enumerate*}[label=\emph{(\roman*)}]
\item its \blockdag~$\G$ in
  lines~\ref{lin:gossip:rcv-b}--\ref{lin:gossip:fwd}, and
\item its current block $\bB$ by including requests and references to
  other blocks in
  lines~\ref{lin:gossip:disseminate}--\ref{lin:gossip:next-b}.
\end{enumerate*}
The servers communicate by exchanging blocks.
Assumption~\ref{ass:reliable-links} guarantees, that a correct~$s$
will eventually receive a block from another correct server. Moreover,
every correct server~$s$ will regularly request $\disseminate()$ in
line~\ref{lin:gossip:disseminate} and will eventually send their own
block~$\bB$ in line~\ref{lin:gossip:brdcst}. This is guaranteed by the
high-level protocol~(\cf~Section~\ref{sec:interfacing}).

Every server~$s$ operates on four data structures. The two data
structures which are shared with Algorithm~\ref{alg:interpret} are
given as arguments in line~\ref{lin:gossip}:
\begin{enumerate*}[label=\emph{(\roman*)}]
\item the \blockdag~$\G$, which Algorithm~\ref{alg:interpret} will
  only read, and
\item a buffer~$\requestBuffer$, where Algorithm~\ref{alg:interpret}
  inserts pairs of labels and requests.
  On the other hand, $s$ also keeps
\item the block~$\bB$ which $s$ currently builds
  (line~\ref{lin:gossip:bb}), and
\item a buffer~$\buffer$ of received blocks
  (line~\ref{lin:gossip:buffer}).
\end{enumerate*}
To build its \blockdag, $s$ inserts blocks into $\G$ in
line~\ref{lin:gossip:ins-b-g} and line~\ref{lin:gossip:insB}. It is
guaranteed that by inserting those blocks $\G$ remains a \blockdag
\citeAppendix{Lemma~\ref{lem:g-is-block-dag}}{Lemma~A.5}.  To insert a
block, $s$ keeps track of its received blocks as candidate blocks in
the buffer~$\buffer$
(line~\ref{lin:gossip:rcv-b}--\ref{lin:gossip:ins-b-buffer}).
Whenever $s$ considers a $B' \in \buffer$
valid~(line~\ref{lin:gossip:valid}), $s$ inserts $B'$ in $\G$
(line~\ref{lin:gossip:ins-b-g}).
However, to consider a block~$B'$ valid, $s$ has to consider all its
predecessors valid---and $s$ may not have yet received every
$B \in B'.\preds$. That is, $B' \in \buffer$ but $B \not\in \buffer$
and $B \not\in \G$ (\cmp line~\ref{lin:gossip:missing-preds}).
Now, $s$ can request forwarding of $B$ from the server that built~$B'$,
\ie from $s'$ where $B'.\nid = s'$, by sending $\FWD\ B$ to $s'$
(lines~\ref{lin:gossip:deltaB}--\ref{lin:gossip:send-fwd}).
To prevent $s$ from flooding $s'$ an implementation would guard
lines~\ref{lin:gossip:missing-preds}--\ref{lin:gossip:send-fwd}, \eg
by a timer~$\timer{B'}$.
That is, we implicitly assume that for every block~$B'$ a correct
server waits a reasonable amount of time before (re-)issuing a forward
request. The wait time should be informed by the estimated round-trip
time and can be adapted for repeating forwarding requests.

On the other hand, $s$ also answers to forwarding requests for a
block~$B$ from $s'$, where $B \in B'.\preds$ of some block $B'$
disseminated by $s$
(lines~\ref{lin:gossip:rcv-fwd}--\ref{lin:gossip:fwd}). It is not
necessary to request forwarding from servers other than $s'$. We only
require that correct servers will eventually share the same blocks.
This mechanism, together with Assumption~\ref{ass:reliable-links} and
$s$'s eventual dissemination of $\bB$, allows us to establish the
following lemma:

\begin{lemma} \label{lem:eventually} % %
  For a correct server~$s$ executing $\gossip$, if $s$ receives a
  block $B$, which $s$ considers valid, then
  \begin{enumerate*}
  \item \label{lem:eventually-receive} %
    every correct server will eventually receive $B$, and
  \item \label{lem:eventually-valid} %
    every correct server will eventually consider $B$ valid.
  \end{enumerate*}
\end{lemma}
%
% Intution of Lemma~\ref{lem:eventually}~(\ref{lem:eventually-receive}):
% $s$ has received $B$ and by disseminating a block~$B'$ with
% $B'.\nid = s$ where $B \in B'.\preds$, $s$ promises to every other
% server~$s'$ to forward $B$ if necessary. As $s'$ will receive~$B'$
% eventually, $s'$ at the latest then learns of $B$'s existence and can
% get $B$ from $s$. Part (\ref{lem:eventually-valid}) of
% Lemma~\ref{lem:eventually} is a straight-forward consequence of
% part~(\ref{lem:eventually-receive}).

In parallel to building~$\G$, $s$ builds its current block $\bB$ by
\begin{enumerate*}[label=\emph{(\roman*)}]
\item continuously adding a reference to any block~$B'$, which $s$
  receives and considers valid in line~\ref{lin:gossip:ins-pred}
  (adding at most one reference to $B'$ \citeAppendix{(
    Lemma~\ref{lem:ref-at-most-once})}{Lemma~A.6}), and
\item eventually sending~$\bB$ to every server in
  line~\ref{lin:gossip:brdcst}.
\end{enumerate*}
Just before $s$ sends $\bB$, $s$ injects literal inscriptions of
$(\ell_i, r_i) \in \requestBuffer$ into $\bB$ in
line~\ref{lin:gossip:sign}. Now $\rs$ holds requests~$r_i$ for the
protocol instances~$\P$ with label~$\ell_i$. These requests will
eventually be read in Algorithm~\ref{alg:interpret}.
Finally, $s$ signs $\bB$ in line~\ref{lin:gossip:sign}, sends $\bB$ to
every server, and starts building its next $\bB$ in
line~\ref{lin:gossip:next-b} by incrementing the sequence
number~$\seqn$, initializing $\preds$ with the parent block, as well
as clearing $\rs$ and $\signature$.

%\subsection{Joint \blockdags}

So far we established, how $s$ builds its own \blockdag. Next we want
to establish the concept of a \emph{joint \blockdag} between two
correct servers~$s$ and $s'$. Let $\G_s$ and $\G_{s'}$ be the
\blockdag of $s$ and $s'$. We define their \emph{joint
  \blockdag}~$\G'$ as a \blockdag $\G' \geqslant \G_s \cup
\G_{s'}$. This joint \blockdag is a \blockdag for $s$ and for $s'$
\citeAppendix{(Lemma~\ref{lem:cupdag})}{Lemma~A.7}.
Intuitively, we want any two correct servers to be able to `gossip
some more' and arrive at their joint \blockdag~$\G'$.

\begin{lemma} \label{th:joint-dag} %
  Let $s$ and $s'$ be correct servers with \blockdags $\G_{s}$ and
  $\G_{s'}$.  By executing $\gossip$ in Algorithm~\ref{alg:gossip},
  eventually $s$ has a \blockdag~$\G'_s$ such that
  $\G_s' \geqslant {\G_s \cup \G_{s'}}$.
\end{lemma}

\begin{proof}
  By \citeAppendix{Lemma~\ref{lem:g-is-block-dag}}{Lemma~A.5} any \blockdag~$\G'$ obtained
  through $\gossip$ is a \blockdag, and by \citeAppendix{Lemma~\ref{lem:cupdag}}{Lemma~A.7}
  $\G'$ is a \blockdag for $s$.
  It remains to show that by executing $\gossip$, eventually $\G'$
  will be the \blockdag for $s$. As $s'$ received and considers all
  $B \in \G_{s'}$ valid, by
  Lemma~\ref{lem:eventually}~(\ref{lem:eventually-valid}) $s$ will
  eventually consider every $B$ valid. By executing $\gossip$, $s$
  will eventually insert every $B$ in its \blockdag and $\G'$ will
  contain all $B \in \G_{s'}$.
\end{proof}

In the next section, we will show how $s$ and $s'$ can independently
interpret a deterministic protocol~$\P$ on this joint \blockdag. But
before we do so, we want to highlight that the $\gossip$ protocol
retains the key benefits reported by works using the \blockdag
approach, namely simplicity and amenability to high-performance
implementation.
Currently, our $\gossip$ protocol in Algorithm~\ref{alg:gossip} uses an
explicit forwarding mechanism in
lines~\ref{lin:gossip:missing-preds}--\ref{lin:gossip:fwd}. This
explicit forwarding mechanism---as opposed to every correct server
re-transmitting \emph{every} received and valid block in a second
communication round---is possible through blocks including references
to predecessor blocks. Hence, every server knows what blocks it is
missing and whom to ask for them. But in an implementation, we would
go a step further and replace the forwarding mechanism---and
messages---as described next:
each block is associated with a unique cryptographic reference that
authenticates its content. As a result both best-effort broadcast
operations as well as synchronization operations can be implemented
using distributed and scalable key-value stores at each server (\eg,
\cassandra, \awsse), which through sharding have no limits on their
throughput. Best-effort broadcasts can be implemented directly,
through simple asynchronous IO.
This is due to the the (now) single type of message, namely blocks,
and a single handler for blocks in gossip that performs minimal work:
it just records blocks, and then asynchronously ensures their
predecessors exist (potentially doing a remote key-value read) and
they are valid (which only involves reference lookups into a
hash-table and a single signature verification).
Alternatively, best-effort broadcast itself can be implemented using a
publish-subscribe notification system and remote reads into
distributed key value stores.
In summary, the simplicity and regularity of $\gossip$, and the weak
assumptions and light processing allow systems' engineers great
freedom to optimize and distribute a server's operations. Both
\hashgraph and \blockmania (which have seen commercial implementation)
report throughputs of many 100,000 transactions per second, and
latency in the order of seconds.
% For hashgraph:  https://hedera.com/hh_whitepaper_v2.1-20200815.pdf
As we will see in the next section no matter which $\P$ the servers $s$ and $s'$
choose to interpret, they can build a joint \blockdag using the same gossip logic---by only
exchanging blocks---and then \emph{independently} interpreting $\P$ on
$\G$.

%%% Local Variables:
%%% mode: latex
%%% TeX-master: "paper"
%%% End:

\section{Interpreting a Protocol} \label{sec:interpret}

Every server~$s$ interprets the protocol~$\P$ embedded in its local
\blockdag~$\G$. This interpretation is completely \emph{decoupled}
from building the \blockdag in Algorithm~\ref{alg:gossip}. To
interpret one \emph{protocol instance} of $\P$ tagged with
label~$\ell$, server~$s$ locally runs one \emph{process instance}
of~$\P$ with label~$\ell$ for every other server $s_i \in
\srvrs$. Thereby, $s$ treats $\P$ as a black-box which
\begin{enumerate*}[label=\emph{(\roman*)}]
\item takes a request or a message, and
\item returns messages or an indication.
\end{enumerate*}
A server $s$ can fully simulate the protocol instance $\P$ for any
other server because their requests and messages have been embedded in
the \blockdag~$\G$ by Algorithm~\ref{alg:gossip}.
User requests~$r_j$ to $\P$ are embedded in a block $B \in \G$ in
$B.\rs$ and $s$ reads these requests from the block and passes them
on to the simulation of $\P$. Since $\P$ is deterministic, $s$
can---after the initial request~$r_j$ for $\P$---compute all
subsequent messages which would have been sent in $\P$ by interpreting
edges between blocks, such as $B_1 \nxt B_2$, as messages sent from
$B_1.\nid$ to $B_2.\nid$.  There is no need for explicitly sending
these messages.  And indeed, our goal is to show that the
interpretation of a deterministic protocol $\P$ embedded in a
\blockdag implements a reliable point-to-point link.
%\subsection{Blackbox $\P$}

To treat $\P$ as a black-box, we assume the following high-level
interface:
\begin{enumerate*}[label=\emph{(\roman*)}]
\item an interface to \emph{request} $r \in \rqsts_\P$, and
\item an interface where $\P$ \emph{indicates} $i \in \inds_\P$.
\end{enumerate*}
When a request~$r$ reaches a process instance, we assume that it
immediately returns messages $m_1, \ldots, m_k$ triggered by $r$.  This is justified,
as $s$ runs all process instances locally. As requests do not depend
on the state of the process instance, also these messages do not
depend on the current state of process instance.
We also assume a low-level interface for $\P$ to \emph{receive} a
message~$m$. Again, we assume that when $m$ reaches a process
instance, it immediately returns the messages $m_1, \ldots, m_k$
triggered by~$m$.

% \subsection{Algorithm~\ref{alg:interpret}}

Algorithm~\ref{alg:interpret} shows the protocol executed by $s$ for
interpreting a deterministic protocol~$\P$ on a \blockdag~$\G$.
The key task is to `get messages from one block and give them to the
next block'.
\begin{algorithm}[t]
  \setcounter{AlgoLine}{0}
  \SubAlgo{%
    \Module $\interpret(\G \in \blkdags, \P \in \Module)$} %
  {
    \smallskip
    $\interpreted[B \in \blks] \assign \false \in \bool$ \label{lin:interpret:interpreted} \\
    \smallskip
    \SubAlgo{ \Upon $B \in \G$ \Where
      $\eligible(B)$ \label{lin:interpret:loop} %
    } %
    { %
      \smallskip
      $B.\pis \assign \Copy\ {B.\parent.\pis} $ \label{lin:interpret:copy-parent} \\
      \smallskip %
      \SubAlgo{ %
        \For \Every $(\ell_j \in \lbls, r_j \in \rqsts) \in
        B.\rs$ \label{lin:interpret:l} } %
      { %
        \smallskip
        $B.\bfr[\Out, \ell_j] \assign B.\pis[\ell_j].r_j$ \label{lin:interpret:initbfr} \\
      }
      \smallskip %
      \SubAlgo{ %
        \For \Every
        $\ell_j \in \{ \ell_j \mid {(\ell_j, r_j) \in B_j.\rs} \land
        {B_j \in \G} \land {B_j \nxtT B}
        \}$ \label{lin:interpret:all-pis} %
      } %
      { %
        \smallskip %
        \SubAlgo{ %
          \For \Every $B_i \in B.\preds$ \label{lin:interpret:preds}
        } %
        { %
          \smallskip %
          $B.\bfr[\In, \ell_j] \assign { B.\bfr[\In, \ell_j] \cup {}}$%
          $ \{ m \mid m \in B_i.\bfr[\Out, \ell_j] \text{ and } m.\receiver
            = B.\nid \}$ \label{lin:interpret:in} \\ %
        }
        \smallskip
        \SubAlgo{\For \Every $m \in B.\bfr[\In, \ell_j]$ \textbf{ordered
            by}
          $<_\msgs$ \label{lin:handle:foreach}} %
        { %
          \smallskip
          $B.\bfr[\Out, \ell_j] \assign {B.\bfr[\Out, \ell_j] \cup
          B.\pis[\ell_j].\receive(m)}$ \label{lin:interpret:bfr} \label{lin:handle:receive} \\
        } %
      }
      \smallskip %
      $\interpreted[B] = \true$ \label{lin:interpreted-b} \\ %
    } %
    \smallskip %
    \SubAlgo{ %
      \Upon $B.\pis[\ell_j].i$
      \label{lin:interpret:indicate} %
    } %
    { %
      \smallskip $\indicate(\ell_j, i,
      B.\nid)$ \label{lin:interpret:indicate-up} } %
  } %
  \caption{Interpreting protocol~$\P$ on the \blockdag~$\G$.}
  \label{alg:interpret}
\end{algorithm}
Therefore $s$ traverses through every $B \in \G$. To keep track of
which blocks in $\G$ it has already interpreted, $s$ uses
$\interpreted$ in line~\ref{lin:interpret:interpreted}. Note, that
edges in $\G$ impose a partial order: $s$ considers a block $B \in \G$
as $\eligible(B)$ for interpretation if
\begin{enumerate*}[label=\emph{(\roman*)}]
\item $\interpreted[B] = \false$, and
\item for every $B_i \in B.\preds$,
  $\interpreted[B_i] = \true$ holds.
\end{enumerate*}
While there may be more than one $B$ eligible, every $B \in \G$ is
interpreted eventually \citeAppendix{(Lemma~\ref{lem:picked-eventually})}{Lemma~A.10}.
Now $s$ picks an eligible $B$ in line~\ref{lin:interpret:loop} and
\emph{interprets} $B$ in
lines~\ref{lin:interpret:copy-parent}--\ref{lin:interpreted-b}. To
interpret~$B$, $s$ needs to keeps track of two variables for every
protocol instance~$\ell_j$:
\begin{enumerate*}
\item the state of the process instance~$\ell_j$ for a
  server~$s_i \in \srvrs$ in $\pis[\ell_j]$, and
\item the state of $\In$-going and $\Out$-going messages in
  $\bfr[\In, \ell_j]$ and $\bfr[\Out, \ell_j]$.
\end{enumerate*}

Our goal is to track changes to these two variables---the process
instances~$\pis$ and message buffers~$\bfr$---throughout the
interpretation of $\G$. To do so, we assign their state to every
block~$B$. Before $B$ is interpreted, we assume $B.\pis[\ell_j]$ to be
initialized with $\bot$, and $B.\bfr[d \in [\{ \In, \Out \}, \ell_j]$
with $\varnothing$. They remain so while $B$ is $\eligible$
\citeAppendix{(Lemma~\ref{lem:pis:is-init})}{Lemma~A.15}.

After interpreting $B$,
\begin{enumerate*}
\item $B.\pis[\ell_j]$ holds the state of the process
  instance~$\ell_j$ of the server~$s_i$, which built $B$, \ie,
  $s_i = B.\nid$, and
\item $B.\bfr[\In, \ell_j]$ holds the in-going messages for
  $s_i$ and $\bfr[\Out, \ell_j]$ the out-going messages from $s_i$ for
  process instance~$\ell_j$\footnote{An equivalent representation
    would keep process instances~$\pis[B, \ell_j, B.\nid]$ and message
    buffers~$\bfr[B, d \in \{ \In, \Out \}, \ell_j]$ explicitly as
    global state. We chose this notation to accentuate the information
    flow throughout $\G$.}.
\end{enumerate*}

As a starting point for computing the state of $B.\pis[\ell_j]$, $s$
copies the state from the parent block of $B$ in
line~\ref{lin:interpret:copy-parent}. For the base case, \ie all
(genesis) blocks~$B$ without parents, we assume
$B.\pis[\ell_j] \assign \New\ \Process\ \P(\ell_j, s_i)$ where
$s_i = B.\nid$. This is effectively a simplification: we assume a
running process instance~$\ell_j$ for every $s_i \in \srvrs$. In an
implementation, we would only start process instances for $\ell_j$
after receiving the first message or request for $\ell_j$ for
$s_i = B.\nid$. Now in our simplification, we start all process
instances for every label at the genesis blocks and pass them on from
the parent blocks. This leads us to our step case: $B$ has a
parent. As $B.\parent \in B.\preds$, $B.\parent$ has been interpreted
and moreover $B.\parent.\nid = s_i$ \citeAppendix{(Lemma~\ref{lem:pis:nid})}{Lemma~A.13}.
Next, to advance the copied state on $B$, $s$ processes
\begin{enumerate*}
\item \label{it:req} %
  all incoming requests~$r_j$ given by $B.\rs$ in
  lines~\ref{lin:interpret:l}--\ref{lin:interpret:initbfr}, and
\item \label{it:msg} %
  all incoming messages from $B_i.\nid$ to $B.\nid$ given by
  $B_i \nxt B$ in
  lines~\ref{lin:interpret:preds}--\ref{lin:interpret:bfr}.
\end{enumerate*}
For the former (\ref{it:req}), $s$ reads the labels and requests from
the field $B.\rs$. Here $r_j$ is the literal transcription of the
user's original request given to $\P$. To give an example, if $\P$
is reliable broadcast, then $r_j$ could read `$\broadcast(42)$' (\cf
Section~\ref{sec:interfacing}). When interpreting, $s$ requests $r_j$
from $B.\nid$'s simulated protocol instance: $B.\pis[\ell_j].r_j$.
For the latter~(\ref{it:msg}), $s$ collects
\begin{enumerate*}[label=\emph{(\roman*)}]
\item in $B.\bfr[\In, \ell]$ all messages for $B.\nid$ from
  $B_i.\bfr[\Out, \ell]$ where $B_i \in B.\preds$ in
  lines~\ref{lin:interpret:preds}--\ref{lin:interpret:in} and then
  feeds
\item $m \in B.\bfr[\In, \ell]$ to $B.\pis[\ell]$ in
  lines~\ref{lin:handle:foreach}--\ref{lin:handle:receive} in order
  $<_\msgs$. This (arbitrary) order is a simple way to guarantee that
  every server interpreting Algorithm~\ref{alg:interpret} will execute
  exactly the same steps.
\end{enumerate*}
By feeding those messages and requests to $B.\pis[\ell_j]$ in
lines~\ref{lin:interpret:initbfr} and \ref{lin:interpret:bfr} $s$
computes
\begin{enumerate*}
\item the next state of $B.\pis[\ell_j]$ and
\item the out-going messages from $B.\nid$ in $B.\bfr[\Out,
  \ell_j]$. By construction, $m.\sender = B.\nid$ for
  $m \in B.\bfr[\Out, \ell_j]$
  \citeAppendix{(Lemma~\ref{lem:out-sender})}{Lemma~A.14}.
\end{enumerate*}
Once, $s$ has completed this, $s$ marks $B$ as interpreted in
line~\ref{lin:interpreted-b} and can move on to the next eligible
block.
After $s$ interpreted $B$, the simulated process instance
$B.\pis[\ell_j]$ may indicate $i \in \inds$. If this is the case, $s$
indicates~$i$ for $\ell_j$ on behalf of $B.\nid$ in
lines~\ref{lin:interpret:indicate}--\ref{lin:interpret:indicate-up}.
Note, that none of the steps used the fact that it was $s$ who
interpreted $B \in \G$. So, for every $B$, every $s' \in \srvrs$ will
come to the exact same conclusion.

But we glossed over a detail, $s$ actually had to take a
choice---more than one $B$ may have been eligible in
line~\ref{lin:interpret:loop}. This is a feature: by having this
choice we can think of interpreting a $\G'$ with $\G' \geqslant \G$ as
an `extension' of interpreting $\G$.
And, for two eligible $B_1$ and $B_2$ it does not matter if we pick
$B_1$ before $B_2$. Informally, this is because when we pick $B_1$ in
line~\ref{lin:interpret:loop}, only the the state with respect to
$B_1$ is modified---and this state does not depend on $B_2$
\citeAppendix{(Lemma~\ref{lem:unchanged})}{Lemma~A.11}.
Another detail we glossed over is line~\ref{lin:interpret:all-pis}:
when interpreting $B$, $s$ interprets the process instances of every
$\ell_j$ relevant on $B$ \emph{at the same time}. But again, because
$\ell_j \neq \ell'_j$ are independent instances of the protocol with
disjoint messages, \ie, $B_i.\bfr[\Out, \ell_j]$ in
line~\ref{lin:interpret:in} is independent of any
$B_i.\bfr[\Out, \ell'_{j}]$, they do not influence each other and the
order in which we process $\ell_j$ does not matter.

Finally, we give some intuition on how byzantine servers can influence
$\G$ and thus the interpretation of $\P$. When running $\gossip$, a
byzantine server $\byz{s}$ can only manipulate the state of $\G$ by
\begin{enumerate*}
\item \label{it:equiv} sending an equivocating block, \ie building a
  $B$ and $B'$ with $\byz{s} = B.\parent.\nid$ and
  $\byz{s} = B'.\parent.\nid$. When interpreting $B$ and $B'$, $s$
  will split the state for $\byz{s}$ and have two `versions' of
  $\pis[\ell_j]$---$B'.\pis[\ell_j]$ and $B.\pis[\ell_j]$---sending
  conflicting messages for $\ell_j$ to servers referencing $B$ and
  $B'$. But as $\P$ is a BFT protocol, the servers $s_i$ simulating
  $\P$ (run by $s$) can deal with equivocation. Then $\byz{s}$ could
\item \label{it:twice} reference a block multiple times, or
\item \label{it:sil} never reference a block. But again as $\P$ is a
  BFT protocol, the servers $s_i$ simulating $\P$ can deal with
  duplicate messages and with silent servers.
\end{enumerate*}

% \subsection{Communication channel}

Going back to Algorithm~\ref{alg:interpret}, the key task of $s$
interpreting $\G$ is to get messages from one block to the next
block. So we can see this interpretation of a \blockdag as an
implementation of a \emph{communication channel}. That is,
for a correct server~$s$ executing $s.\interpret(\G, \P)$
\begin{enumerate*}[label=\emph{(\roman*)}]
\item a server~$s_1$ \emph{sends} messages $m_1, \ldots, m_k$ for a
  protocol instance $\ell_j$ in either
  line~\ref{lin:interpret:initbfr} or line~\ref{lin:interpret:bfr} of
  Algorithm~\ref{alg:interpret}, and
\item a server $s_2$ \emph{receives} a message~$m$ for a protocol
  instance~$\ell_j$ in line~\ref{lin:handle:receive} of
  Algorithm~\ref{alg:interpret}.
\end{enumerate*}
The next lemma relates the sent and received messages with the message
buffers $\bfr$ and follows from tracing changes to the variables in
Algorithm~\ref{alg:interpret}:

\begin{lemma} \label{lem:bfr}
  For a correct server~$s$ executing $s.\interpret(\G, \P)$
  \begin{enumerate}
  \item \label{lem:send-bfr} %
    a server~$s_1$ \emph{sends} $m$ for a protocol instance $\ell'$
    iff
    there is a $B_1 \in \G$ with $B_1.\nid = s_1$ such that
    $m \in B_1.\bfr[\Out, \ell']$ for a $B' \in \G$ with
    $(\ell', r) \in B'.\rs$ and $B' \nxtS B_1$.
  \item \label{lem:receive-bfr}%
    a server~$s_2$ \emph{receives} a message~$m$ for protocol instance
    $\ell'$ iff there are some $B_1, B_2 \in \G$ with $B_1 \nxt B_2$
    and $B_2.\nid = s_2$ and $m \in B_2.\bfr[\In, \ell']$ for a
    $B' \in \G$ such that $(\ell', r) \in B'.\rs$ and $B' \nxtS B_1$.
  \end{enumerate}
\end{lemma}

The following lemma shows our key observation from before:
interpreting a \blockdag is independent from the server doing the
interpretation. That is, $s$ and $s'$ will arrive at the same state
when interpreting $B \in \G$.

\begin{lemma} \label{lem:joint-eq} %
  If $\G \leqslant \G'$then for every $B \in \G$, a
  deterministic protocol~$\P$ and correct servers~$s$ and $s'$
  executing $s.\interpret(\G, \P)$ and $s'.\interpret(\G', \P)$ it holds that
  $B.\pis[\ell_j] = B.\pis'[\ell_j]$ and
  $B.\bfr[\Out, \ell_j] = B.\bfr'[\Out, \ell_j]$ for
  $(\ell_j, r) \in B_j.\rs$ with $B_j \nxtN{n} B$ for $n \geqslant 0$.
\end{lemma}

\begin{proof}
  In this proof, when executing $s'.\interpret(\G', \P)$ we write
  $\bfr'$ and $\pis'$ to distinguish from $\bfr$ and $\pis$ when
  executing $s.\interpret(\G, \P)$.
  We show $B_1.\bfr[\Out, \ell_j] = B_1.\bfr'[\Out, \ell_j]$ and
  $B_1.\pis[\ell_j] = B_1.\pis'[\ell_j]$ by induction on $n$---the
  length of the path from $B_j$ to $B_1$ in $\G$ and $\G'$.
  For the base case we have $B_1 = B_j$ and
  $ \ell_j \in \{ \ell_j \mid (\ell_j, r_j) \in B_1.\rs \}$. By
  \citeAppendix{Lemma~\ref{lem:picked-eventually}}{Lemma~A.10}, $B_1$
  is picked eventually in line~\ref{lin:interpret:loop} of
  Algorithm~\ref{alg:interpret} when executing $s.\interpret(\G,
  \P)$. Then, by line~\ref{lin:interpret:initbfr}
  $B_1.\bfr[\Out, \ell]$ is $B_1.\pis[\ell_j].(B_1.\rs)$.
  By the same reasoning, when executing $s'.\interpret(\G', \P)$,
  $B_1.\bfr'[\Out, \ell] = B_1.\pis[\ell_j].(B_1.\rs)$. As
  $B_1.\pis[\ell_j].(B_1.\rs)$ are deterministic and depend only on
  $B_1$, $\ell_j$, and $\P$, we know that
  $B_1.\pis[\ell] = B_1.\pis'[\ell]$ and
  $B_1.\pis[\ell] = B_1.\pis'[\ell]$, and conclude the base case.
  For the step case by induction hypothesis for $B_i \in B_1.\preds$
  with $B_j \nxtN{n-1} B_i$ holds
  \begin{enumerate*}[label=\emph{(\roman*)}]
  \item \label{it:bfr} $B_i.\bfr[\Out, \ell_j] = B_i.\bfr'[\Out, \ell_j]$, and
  \item \label{it:pis} $B_i.\pis[\ell_j] = B_i.\pis'[\ell_j]$.
  \end{enumerate*}
  Again by
  \citeAppendix{Lemma~\ref{lem:picked-eventually}}{Lemma~A.10}, $B_1$
  is picked eventually in line~\ref{lin:interpret:loop} of
  Algorithm~\ref{alg:interpret} when executing $s.\interpret(\G, \P)$
  and $s'.\interpret(\G', \P)$.
  In line~\ref{lin:interpret:copy-parent} and as
  $B_1.\parent \in B_1.\preds$ and \ref{it:pis}, now
  $B_1.\pis[\ell_j] = B_1.\pis'[\ell_j]$.
  Now, as $\P$ is deterministic, we only need to establish that
  $B_1.\bfr[\In, \ell_j] = B_1.\bfr'[\In, \ell_j]$ to conclude that
  $B_1.\pis[\ell_j] = B_1.\pis'[\ell_j]$ and
  $B_1.\bfr[\Out, \ell_j] = B_1.\bfr'[\Out, \ell_j]$, which as
  $ (\ell_j, r) \not\in B_1.\rs$, is only modified in this
  line~\ref{lin:interpret:bfr}.
  By \citeAppendix{Lemma~\ref{lem:uninterpreted}}{Lemma~A.9}, we know
  for both executions that
  $B_1.\bfr[\In, \ell_j] = B_1.\bfr'[\In, \ell_j] = \varnothing$,
  before $B_1$ is picked. Now, by \ref{it:bfr} and
  line~\ref{lin:interpret:in}
  $B_1.\bfr[\In, \ell_j] = B_1.\bfr'[\In, \ell_j]$, and we conclude
  the proof.
\end{proof}

A straightforward consequence of Lemma~\ref{lem:joint-eq} is, that
when in the interpretation of $s$, a server $s_1$ sends a message~$m$
for $\ell_j$, then $s_1$ sends $m$ in the interpretation of $s'$
\citeAppendix{(Lemma~\ref{lem:send-transfer})}{Lemma~A.16}. Curiously,
$s_1$ does not have to be correct: we know $s_1$ sent a block~$B$ in
$\G$, that corresponds to a message~$m$ in the interpretation of
$s$. Now this block will be interpreted by $s'$ and the same message
will be interpreted---and for that the server~$s_1$ does not need to
be correct.
By Lemma~\ref{lem:point-to-point} $\interpret(\G, \P)$ has the
properties of an authenticated perfect point-to-point link after
\cite[Module 2.5, p.\ 42]{2011_cachin_et_al}.

\begin{lemma} \label{lem:point-to-point} %
  For a \blockdag~$\G$ and a correct server~$s$ executing
  $s.\interpret(\G, \P)$ holds
 \begin{enumerate}
 \item \label{al:rd} if a correct server~$s_1$ sends a message $m$ for
   a protocol instance~$\ell$ to a correct server~$s_2$, then $s_2$
   eventually receives $m$ for protocol instance~$\ell$ for a correct
   server~$s'$ executing $s'.\interpret(\G', \P)$ and a \blockdag
   $\G' \geqslant \G$ (\emph{reliable delivery}).
 \item \label{al:nd} for a protocol instance~$\ell$ no message is
   received by a correct server~$s_2$ more than once (\emph{no duplication}).
 \item \label{al:at} if some correct server~$s_2$ receives a
   message~$m$ for protocol instance~$\ell$ with sender $s_1$ and
   $s_1$ is correct, then the message~$m$ for protocol instance~$\ell$
   was previously sent to $s_2$ by $s_1$ (\emph{authenticity}).
  \end{enumerate}
\end{lemma}

\begin{proof}[Proof Sketch]
  For (\ref{al:rd}), we observe that every message~$m$ sent in
  $s.\interpret(\G, \P)$ will be sent in $s'.\interpret(\G', \P)$ for
  $\G' \geqslant \G$ by
  \citeAppendix{Lemma~\ref{lem:send-transfer}}{Lemma~A.16}.  Now by
  Lemma~\ref{th:joint-dag}, $s'$ will eventually have some
  $\G' \geqslant \G$. By Lemma~\ref{lem:bfr}~(\ref{lem:send-bfr}) we
  have witnesses $B_1, B_2 \in \G'$ with $B_1 \nxt B_2$, and by
  Lemma~\ref{lem:bfr}~(\ref{lem:receive-bfr}) we found a witness $B_2$
  to receive the message on when executing $s'.\interpret(\G', \P)$.
  For (\ref{al:nd}), we observe, that duplicate messages are only
  possible if $s_2$ inserted the block~$B_1$, which gives rise to the
  message~$m$, in two different blocks built by $s_2$. But this
  contradicts the correctness of $s_2$ by
  \citeAppendix{Lemma~\ref{lem:ref-at-most-once}}{Lemma~A.6}. For
  (\ref{al:at}), we observe that only $s_1$ can build and sign any
  block $B_1$ with $s_1 = B.\nid$, which gives rise to $m$.
\end{proof}

Before we compose $\gossip$ and $\interpret$ in the next section under
a $\shim$, we highlight the key benefits of using $\interpret$ in
Algorithm~\ref{alg:interpret}.
By leveraging the \blockdag structure together with $\P$'s
determinism, we can \emph{compress messages} to the point of omitting
some of them. When looking at line~\ref{lin:handle:receive} of
Algorithm~\ref{alg:interpret}, the messages in the
buffers~$\bfr[\Out, \ell]$ and $\bfr[\In, \ell]$ have never been sent
over the network. They are locally computed, functional results of the
calls $\receive(m)$. The only `messages' actually sent over the
network are the requests~$r_i$ read from $B.\rs$ in
line~\ref{lin:interpret:initbfr}.
To determine the messages following from these request, the server~$s$
simulates an instance of protocol~$\P$ for every
$s_i \in \srvrs$---simply by simulating the steps in the deterministic
protocol. However, not every step can be simulated: as $s$ does not
know $s_i$'s private key, $s$ cannot sign a message on $s_i$'s
behalf. But then, this is not necessary, because $s$ can derive the
authenticity of the message triggered by a block~$B$ from the
signature of $B$, \ie, $B.\signature$. So instead of signing individual
messages, $s_i$ can give a \emph{batch signature}~$B.\signature$ for
authenticating every message materialized through~$B$.
Finally, $s$ interprets protocol instances with labels~$\ell_j$
\emph{in parallel} in line~\ref{lin:interpret:all-pis} of
Algorithm~\ref{alg:interpret}. While traversing the \blockdag, $s$
uses the structure of the \blockdag to interpret requests and messages
for every~$\ell_j$. Now, the same block giving rise to a request in
process instance~$\ell_j$ may materialize a message in process
instance~$\ell_j'$. The (small) price to pay is the increase of block
size by references to predecessor blocks, \ie, $B.\preds$.
We will illustrate the benefits again on the concrete example of
byzantine reliable broadcast in the next
Section~\ref{sec:interfacing}.

%%% Local Variables:
%%% mode: latex
%%% TeX-master: "paper"
%%% End:

\section{Using the Framework} \label{sec:interfacing}

The protocol $\shim(\P)$ in Algorithm~\ref{alg:shim} is responsible
for the choreography of the external user of $\P$, the $\gossip$
protocol in Algorithm~\ref{alg:gossip}, and the $\interpret$ protocol
in Algorithm~\ref{alg:interpret}.
\begin{algorithm}[t]
  \setcounter{AlgoLine}{0} %
  \SubAlgo{\Module $\shim(s \in \srvrs, \P \in \Module)$} %
  { %
    \smallskip
    $\requestBuffer \assign \emp \in \pow{\lbls \times \rqsts}$ \label{lin:shim:buffer} \\
    $\G \assign \empG \in \blkdags$ \label{lin:gossip:g} \\
    \smallskip
    $\gssp \assign \New\ \Process\ \gossip(s, \G, \requestBuffer)$ \label{lin:shim:init-gossip} \\
    $\intprt \assign \New\ \Process\ \interpret(\G, \P)$ \label{lin:shim:init-interpret} \\
    \smallskip %
    \SubAlgo{ %
      \Upon $\request(\ell \in \lbls, r \in \rqsts)$ \label{lin:shim:auth-r} %
    } %
    { %
      $\requestBuffer.\insertRequest(\ell, r)$ \label{lin:shim:ins-r-buffer} \\
    }
    \smallskip %
    \SubAlgo { %
      \Upon $\intprt.\indicate(\ell, i, s')$ \Where $s' =
      s$ \label{lin:shim:indicate} %
    } %
    { %
      $\indicate(\ell, i)$ \label{lin:shim:indicate-up} %
    } %
    \smallskip %
    \SubAlgo { %
      \textbf{repeatedly} \label{lin:repeatedly} %
    } %
    { %
      $\gssp.\disseminate()$  \label{lin:disseminate} %
    } %
  } %
  \caption{Interfacing between $\gossip$, $\interpret$ and user of
    $\P$.}
  \label{alg:shim}
\end{algorithm}
Therefore, the server~$s$ executing $\shim(\P)$ in
Algorithm~\ref{alg:shim} keeps track of two synchronized data
structures
\begin{enumerate*}
\item a buffer of labels and requests $\requestBuffer$ in
  line~\ref{lin:shim:buffer}, and
\item and the \blockdag $\G$ in line~\ref{lin:gossip:g}.
\end{enumerate*}
By calling $\requestBuffer.\insertRequest(\ell, r)$, $s$ inserts
$(\ell, r)$ in $\requestBuffer$, and by calling
$\requestBuffer.\takeRequests()$, $s$ gets \emph{and} removes a
suitable number of requests $(\ell_1, r_1), \ldots, (\ell_n, r_n)$
from $\requestBuffer$. To insert a block~$B$ in $\G$, $s$ calls
$\G.\ins(B)$ from Definition~\ref{def:block-dag}.
We tacitly assume these operations are atomic.
When starting an instance of $\gossip$ and $\interpret$ in
line~\ref{lin:shim:init-gossip} and~\ref{lin:shim:init-interpret}, $s$
passes in references to theses shared data structures.
When the external user of protocol~$\P$ requests $r \in \rqsts$ for
$\ell \in \lbls$ from $s$ via the request $\request(\ell, r)$ to
$\shim(\P)$ then $s$ inserts $(\ell, r)$ in $\requestBuffer$ in
lines~\ref{lin:shim:auth-r}--\ref{lin:shim:ins-r-buffer}. By executing
$\gossip$, $s$ writes $(\ell, r)$ in $\bB$ in
Algorithm~\ref{alg:gossip}~line~\ref{lin:gossip:inject}, and as
eventually $\bB \in \G$, $r$ will be requested from protocol instance
$\pis[\ell]$ when $s$ executes line~\ref{lin:interpret:initbfr} in
Algorithm~\ref{alg:interpret}
\citeAppendix{(Lemma~\ref{lem:r-in-p})}{Lemma~A.17}. On the other
hand, when $\interpret$ indicates $i \in \inds$, for the
interpretation of $\P$ \emph{for itself}, \ie, $s = s'$, then $s$
indicates to the user of $\P$ in
line~\ref{lin:shim:indicate}--\ref{lin:shim:indicate-up} of
Algorithm~\ref{alg:shim} \citeAppendix{(Lemma~\ref{lem:i-in-shim-p})}{Lemma~A.18}.  For $s$
to only indicate when $s = s'$ might be an over-approximation: $s$
trusts $s$'s interpretation of $\P$ as $s$ is correct for $s$. We
believe this restriction can be lifted (\cf
Section~\ref{sec:conclusion}).
Finally, as promised in Section~\ref{sec:block-dags}, in
lines~\ref{lin:repeatedly}--\ref{lin:disseminate} $s$ repeatedly
requests $\disseminate$ from $\gossip$ to disseminate $\bB$.
Within the control of $s$, the time between calls to $\disseminate$
can be adapted to meet the network assumptions of $\P$ and can be
enforced \eg by an internal timer, the block's payload, or when $s$
falls $n$ blocks behind. For our proofs we only need to guarantee that
a correct $s$ will eventually request $\disseminate$.

Following \cite{2011_cachin_et_al}, a protocol~$\P$ implements an
interface~$\IF$ and has properties~$\PR$, which are shown to hold for
$\P$. For any property, which holds for a protocol~$\P$ and where
the proof of the property relies on the reliable point-to-point
abstraction in Lemma~\ref{lem:point-to-point}, $\PR$ holds for
$\shim(\P)$. Again following \cite{2011_cachin_et_al}, these are the
properties of any algorithm that \emph{uses} the reliable
point-to-point link abstraction.

Taking together what we have established for $\gossip$ in
Section~\ref{sec:block-dags}, \ie that correct servers will eventually
share a joint \blockdag, and that $\interpret$ gives a point-to-point
link between them in Section~\ref{sec:interpret}, for $\shim(\P)$ the
following holds:

\begin{theorem} \label{th:meat} %
  For a correct server $s$ and a deterministic protocol~$\P$, if $\P$
  is an implementation of
  \begin{enumerate*}[label=\emph{(\roman*)}]
  \item \label{it:IF} an interface~$\IF$ with requests~$\rqsts_\P$ and
    indications~$\inds_\P$ using the reliable point-to-point link
    abstraction such that
  \item \label{it:PR} a property~$\PR$ holds,
  \end{enumerate*}
  then $\shim(\P)$ in Algorithm~\ref{alg:shim}
  implements~\ref{it:IF} $\IF$ such that \ref{it:PR}~$\PR$ holds.
\end{theorem}

\begin{proof}
  By \citeAppendix{Lemma~\ref{lem:r-in-p}}{Lemma~A.17} and
  \citeAppendix{Lemma~\ref{lem:i-in-shim-p}}{Lemma~A.18}, \ref{it:IF}
  $\shim(\P)$ implements the interface~$\IF$ of $\rqsts_\P$ and
  $\inds_\P$.
  For \ref{it:PR}, by assumption $\PR$ holds for $\P$ using a reliable
  point-to-point link abstraction.  By
  Lemma~\ref{lem:point-to-point} $s.\interpret(\G, \P)$ implements a
  reliable point-to-point link. As Algorithm~\ref{alg:interpret}
  treats $\P$ as a black-box every $B.\pis[\ell]$ holds an execution of
  $\P$. Assume this execution violates $\PR$. But then an execution of
  $\P$ violates $\PR$ which contradicts the assumption that $\PR$ holds
  for~$\P$.
\end{proof}

Our proof relies on a point-to-point link between two correct servers
and thus we can translate the argument of all safety and liveness
properties, for which their reasoning relies on the point-to-point
link abstraction, to our \blockdag framework.
Because we provide an abstraction, we cannot directly translate
implementation-level properties measuring performance such as latency
or throughput.
They rely on the concrete implementation. Also, as discussed in
Section~\ref{sec:interpret}, properties related to signatures do not
directly translate, because blocks---not messages---are
(batch-)signed.

While our work focuses on correctness, two recent works show that
DAG-based approaches for concrete protocols~$\P$ are efficient and
even optimal: DAG-Rider~\cite{2021_keidar_et_al} implements the
asynchronous Byzantine Atomic Broadcast abstraction and is shown to be
optimal with respect to resilience, amortized communication
complexity, and time.  Different to our work, DAG-Rider relies on
randomness, which is an extension in our setting. On the other hand
Narwhal and Tusk~\cite{2021_danezis_et_al} for BFT consensus do not
rely on randomness and report impressive---also empirically
evaluated---performance gains. Moreover, as argued in
\cite{2021_danezis_et_al}, our approach enjoys two further benefits
for implementations: load balancing, as we do not rely on a single
leader, and equal message size.
Finally, we note that in our setting the complexity measure of calls
to the reliable point-to-point link abstraction is slightly
misleading, because we are optimizing the messages transmitted by the
point-to-point link abstraction.

% \subsection{Broadcast}

\paragraph{$\P \assign$ byzantine reliable broadcast} %
In the remainder of this section, we will sketch how a user may use
the \blockdag framework. Our example for $\P$ is \emph{byzantine
  reliable broadcast}~(BRB)---a protocol underlying recently-proposed
efficient payment
systems~\cite{2019_guerraoui_et_al,2020_baudet_et_al}. Given an
implementation of byzantine reliable broadcast after \cite[Module
3.12, p.\ 117]{2011_cachin_et_al}, \eg,
\citeAppendix{Algorithm~\ref{alg:brdcst} in the
  Appendix}{Algorithm~4}: this is the $\P$, which the user passes to
$\shim(\P)$, \ie in the \blockdag framework $\P$ is fixed to an
implementation of BRB, \eg,
\citeAppendix{Algorithm~\ref{alg:brdcst}}{Algorithm~4}. The request in
BRB is $\broadcast(v)$ for a value~$v \in \vals$, so
$\rqsts_\P = \{ \broadcast(v) \mid v \in \vals \}$.
For simplicity and generality, we assume that $\P$---not
$\shim(\P)$---authenticates requests, \ie requests are self-contained
and can be authenticated while simulating $\P$ (\eg,
\citeAppendix{Algorithm~\ref{alg:brdcst}
  line~\ref{lin:brdcst:auth}}{Algorithm~4, line~3}). However, in an
implementation $\shim(\P)$ may be employed to authenticate requests.
On the other hand, BRB indicates with
$\deliver(v)$, so $\inds_\P = \{ \deliver(v) \mid v \in \vals \}$. The
messages sent in BRB are
$\msgs_\P = \{ {\ECHO\ v}, {\READY\ v} \mid v \in \vals \}$ where $\sender$ and
$\receiver$ are the $s \in \srvrs$ running $\shim(\P)$.
When executing line~\ref{lin:interpret:in} of $\interpret(\G, \P)$ in
Algorithm~\ref{alg:interpret}, then $\receive(\ECHO\ 42)$ is
triggered, and \Received $\ECHO\ 42$ holds (\eg, \citeAppendix{in
  Algorithm~\ref{alg:brdcst} in
  line~\ref{lin:brdcst:r-echo}}{Algorithm~4, line~6}). As we assume
$\P$ returns messages immediately, \eg, when the simulation reaches
\Send $\ECHO\ 42$, then $\ECHO\ 42$ is returned immediately (\eg,
\citeAppendix{in line~\ref{lin:brdcst:s-echo} of
  Algorithm~\ref{alg:brdcst}}{Algorithm~4, line 8}).
The interface~$\IF$ is
  $ \rqsts =\{ \broadcast(v) \mid v \in \vals \}$ and
  $\inds = \{ \deliver(v) \mid v \in \vals \}$
  The properties~$\PR$ of BRB---validity, no
  duplication, integrity, consistency, and totality---are preserved.

\begin{figure*}[!t]
  \hspace{-9em}
  \resizebox{0.95\textwidth}{!}{
    \tikzstyle{ipe stylesheet} = [
  ipe import,
  even odd rule,
  line join=round,
  line cap=butt,
  ipe pen normal/.style={line width=0.4},
  ipe pen fat/.style={line width=1.2},
  ipe pen heavier/.style={line width=0.8},
  ipe pen ultrafat/.style={line width=2},
  ipe pen normal,
  ipe mark normal/.style={ipe mark scale=3},
  ipe mark large/.style={ipe mark scale=5},
  ipe mark small/.style={ipe mark scale=2},
  ipe mark tiny/.style={ipe mark scale=1.1},
  ipe mark normal,
  /pgf/arrow keys/.cd,
  ipe arrow normal/.style={scale=7},
  ipe arrow large/.style={scale=10},
  ipe arrow small/.style={scale=5},
  ipe arrow tiny/.style={scale=3},
  ipe arrow normal,
  /tikz/.cd,
  ipe arrows, % update arrows
  <->/.tip = ipe normal,
  ipe dash normal/.style={dash pattern=},
  ipe dash dotted/.style={dash pattern=on 1bp off 3bp},
  ipe dash dash dot dotted/.style={dash pattern=on 4bp off 2bp on 1bp off 2bp on 1bp off 2bp},
  ipe dash dash dotted/.style={dash pattern=on 4bp off 2bp on 1bp off 2bp},
  ipe dash dashed/.style={dash pattern=on 4bp off 4bp},
  ipe dash normal,
  ipe node/.append style={font=\normalsize},
  ipe stretch normal/.style={ipe node stretch=1},
  ipe stretch normal,
  ipe opacity 10/.style={opacity=0.1},
  ipe opacity 30/.style={opacity=0.3},
  ipe opacity 50/.style={opacity=0.5},
  ipe opacity 75/.style={opacity=0.75},
  ipe opacity opaque/.style={opacity=1},
  ipe opacity opaque,
]
\definecolor{red}{rgb}{1,0,0}
\definecolor{blue}{rgb}{0,0,1}
\definecolor{brown}{rgb}{0.647,0.165,0.165}
\definecolor{darkblue}{rgb}{0,0,0.545}
\definecolor{darkcyan}{rgb}{0,0.545,0.545}
\definecolor{darkgray}{rgb}{0.663,0.663,0.663}
\definecolor{darkgreen}{rgb}{0,0.392,0}
\definecolor{darkmagenta}{rgb}{0.545,0,0.545}
\definecolor{darkorange}{rgb}{1,0.549,0}
\definecolor{darkred}{rgb}{0.545,0,0}
\definecolor{gold}{rgb}{1,0.843,0}
\definecolor{gray}{rgb}{0.745,0.745,0.745}
\definecolor{green}{rgb}{0,1,0}
\definecolor{lightblue}{rgb}{0.678,0.847,0.902}
\definecolor{lightcyan}{rgb}{0.878,1,1}
\definecolor{lightgray}{rgb}{0.827,0.827,0.827}
\definecolor{lightgreen}{rgb}{0.565,0.933,0.565}
\definecolor{lightyellow}{rgb}{1,1,0.878}
\definecolor{navy}{rgb}{0,0,0.502}
\definecolor{orange}{rgb}{1,0.647,0}
\definecolor{pink}{rgb}{1,0.753,0.796}
\definecolor{purple}{rgb}{0.627,0.125,0.941}
\definecolor{seagreen}{rgb}{0.18,0.545,0.341}
\definecolor{turquoise}{rgb}{0.251,0.878,0.816}
\definecolor{violet}{rgb}{0.933,0.51,0.933}
\definecolor{yellow}{rgb}{1,1,0}
\definecolor{black}{rgb}{0,0,0}
\definecolor{white}{rgb}{1,1,1}
\begin{tikzpicture}[ipe stylesheet]
  \node[ipe node]
     at (32, 752) {$s_1$};
  \node[ipe node]
     at (32, 680) {$s_2$};
  \draw[{<[ipe arrow tiny]}-]
    (224, 764)
     arc[start angle=36.8699, end angle=143.1301, radius=10];
  \draw
    (48, 768) rectangle (208, 720);
  \node[ipe node]
     at (52, 756) {$B_1$};
  \node[ipe node]
     at (52, 740) {$\ghost{%
\inm = \varnothing}$};
  \node[ipe node]
     at (52, 728) {$\ghost{%
\outm = \ECHO\ 42\ \To\ \{s_1, s_2, s_3, s_4 \}}$};
  \draw[{<[ipe arrow tiny]}-]
    (424, 760)
     arc[start angle=50.1944, end angle=129.8056, radius=31.241];
  \draw[shift={(224, 768)}, xscale=1.1429]
    (0, 0) rectangle (140, -48);
  \node[ipe node]
     at (228, 756) {$B_2$};
  \draw[shift={(224, 704)}, xscale=1.1429]
    (0, 0) rectangle (140, -48);
  \node[ipe node]
     at (228, 692) {$B_{3}$};
  \node[ipe node]
     at (32, 616) {$s_3$};
  \draw[shift={(224, 640)}, xscale=1.1429]
    (0, 0) rectangle (140, -48);
  \node[ipe node]
     at (228, 628) {$B_4
$};
  \node[ipe node]
     at (32, 552) {$s_4$};
  \draw[shift={(224, 576)}, xscale=1.1429]
    (0, 0) rectangle (140, -48);
  \node[ipe node]
     at (228, 564) {$B_5$};
  \draw[shift={(424, 768)}, xscale=1.1429]
    (0, 0) rectangle (140, -48);
  \node[ipe node]
     at (428, 756) {$B_6$};
  \node[ipe node]
     at (428, 740) {$\ghost{%
\inm = \ECHO\ 42\ \From\ \{s_1, s_2, s_3 \}}$};
  \node[ipe node]
     at (428, 728) {$\ghost{%
\outm = \READY\ 42\ \To\ \{s_1, s_2, s_3, s_4 \}}$};
  \draw[shift={(424, 704)}, xscale=1.1429]
    (0, 0) rectangle (140, -48);
  \node[ipe node]
     at (428, 692) {$B_{7}$};
  \draw[shift={(424, 640)}, xscale=1.1429]
    (0, 0) rectangle (140, -48);
  \node[ipe node]
     at (428, 628) {$B_8$};
  \draw[-{>[ipe arrow tiny]}]
    (208, 720)
     -- (224, 640);
  \draw[-{>[ipe arrow tiny]}]
    (208, 720)
     -- (224, 704);
  \draw[-{>[ipe arrow tiny]}]
    (208, 720)
     -- (224, 576);
  \draw[-{>[ipe arrow tiny]}]
    (384, 692)
     -- (424, 740);
  \draw[-{>[ipe arrow tiny]}]
    (380, 640)
     -- (424, 728);
  \draw[{<[ipe arrow tiny]}-]
    (424, 680)
     arc[start angle=50.1944, end angle=129.8056, radius=31.241];
  \draw[{<[ipe arrow tiny]}-]
    (424, 620)
     arc[start angle=50.1944, end angle=129.8056, radius=31.241];
  \draw[-{>[ipe arrow tiny]}]
    (384, 732)
     -- (424, 696);
  \draw[-{>[ipe arrow tiny]}]
    (384, 628)
     -- (424, 668);
  \draw[-{>[ipe arrow tiny]}]
    (384, 552)
     -- (424, 604);
  \draw[-{>[ipe arrow tiny]}]
    (384, 672)
     -- (424, 636);
  \draw[shift={(48, 704)}, xscale=1.1429]
    (0, 0) rectangle (140, -48);
  \draw[shift={(48, 640)}, xscale=1.1429]
    (0, 0) rectangle (140, -48);
  \draw[shift={(48, 576)}, xscale=1.1429]
    (0, 0) rectangle (140, -48);
  \draw[{<[ipe arrow tiny]}-]
    (224, 560)
     arc[start angle=36.8699, end angle=143.1301, radius=10];
  \draw[{<[ipe arrow tiny]}-]
    (224, 620)
     arc[start angle=36.8699, end angle=143.1301, radius=10];
  \draw[{<[ipe arrow tiny]}-]
    (224, 680)
     arc[start angle=36.8699, end angle=143.1301, radius=10];
  \node[ipe node]
     at (112, 776) {$k_1$};
  \node[ipe node]
     at (288, 776) {$k_2$};
  \node[ipe node]
     at (504, 776) {$k_3$};
  \node[ipe node]
     at (228, 740) {$\ghost{%
\inm = \ECHO\ 42\ \From\ \{s_1 \}}$};
  \node[ipe node]
     at (228, 728) {$\ghost{%
\outm = \ECHO\ 42\ \To\ \{s_1, s_2, s_3, s_4 \}}$};
  \node[ipe node]
     at (228, 676) {$\ghost{%
\inm = \ECHO\ 42\ \From\ \{s_1 \}}$};
  \node[ipe node]
     at (228, 664) {$\ghost{%
\outm = \ECHO\ 42\ \To\ \{s_1, s_2, s_3, s_4 \}}$};
  \node[ipe node]
     at (228, 608) {$\ghost{%
\inm = \ECHO\ 42\ \From\ \{s_1 \}}$};
  \node[ipe node]
     at (228, 596) {$\ghost{%
\outm = \ECHO\ 42\ \To\ \{s_1, s_2, s_3, s_4 \}}$};
  \node[ipe node]
     at (228, 544) {$\ghost{%
\inm = \ECHO\ 42\ \From\ \{s_1 \}}$};
  \node[ipe node]
     at (228, 532) {$\ghost{%
\outm = \ECHO\ 42\ \To\ \{s_1, s_2, s_3, s_4 \}}$};
  \node[ipe node]
     at (428, 676) {$\ghost{%
\inm = \ECHO\ 42\ \From\ \{s_1, s_2, s_3 \}}$};
  \node[ipe node]
     at (428, 664) {$\ghost{%
\outm = \READY\ 42\ \To\ \{s_1, s_2, s_3, s_4 \}}$};
  \node[ipe node]
     at (428, 612) {$\ghost{%
\inm = \ECHO\ 42\ \From\ \{s_1, s_2, s_3 \}}$};
  \node[ipe node]
     at (428, 600) {$\ghost{%
\outm = \READY\ 42\ \To\ \{s_1, s_2, s_3, s_4 \}}$};
\end{tikzpicture}
  }
  \caption{The message buffers for process instance $\ell_1$ of a
    \blockdag with $(\ell_1, \broadcast(42)) \in B_1.\rs$}
  \label{fig:block-dag-brdcst}
\end{figure*}
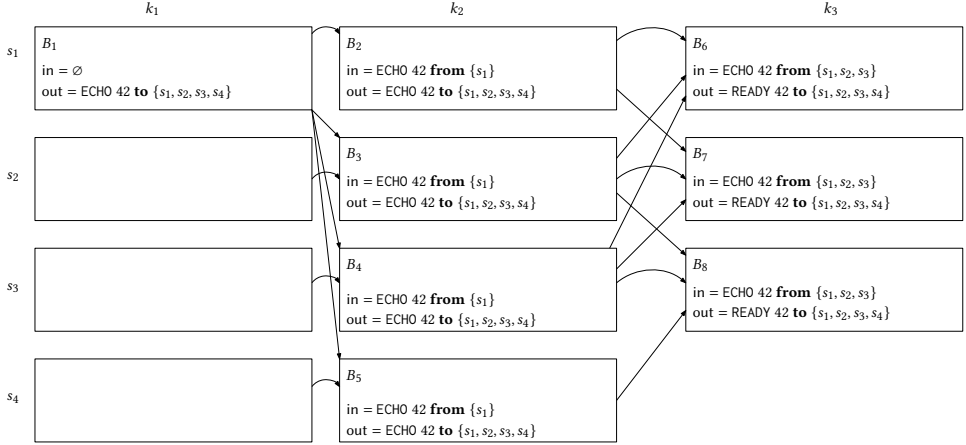

Figure~\ref{fig:block-dag-brdcst} shows a \blockdag for an execution
of $\shim(P)$ using byzantine reliable broadcast. It further
explicitly shows the $\inm$- and $\outm$-going messages from
$\bfr[\In, \ell_1]$ and $\bfr[\Out, \ell_1]$ for a protocol instance
$\ell_1$ and the request $\broadcast(42)$ at block~$B_1$. None of
these messages are ever actually sent over the network---every server
interpreting this \blockdag can use $\interpret$ in
Algorithm~\ref{alg:interpret} to replay an implementation of BRB, \eg
\citeAppendix{Algorithm~\ref{alg:brdcst}}{Algorithm~4}, and get the
same picture.
Figure~\ref{fig:block-dag-brdcst} shows only the (unsent) messages for
$\ell_1$ and $\broadcast(42))$ in $B_1.\rs$, but $B_1.\rs$ may hold
more requests such as $\broadcast(21)$ for $\ell_2$, and all the
messages of all these requests could be materialized in the same
manner---without any messages, or even additional blocks, sent.
And not only $B_1$ holds such requests---also $B_3$ does. For example,
$B_3.\rs$ may contain $\broadcast(25)$ for $\ell_3$. Then, for
$\ell_3$ on $B_3$ materializes $\outm = \ECHO\ 25$ to $s_1$, $s_2$,
$s_3$, and again, without sending any messages, for $\ell_3$ on $B_6$,
$B_7$, and $B_8$ materializes $\inm = \ECHO\ 25$ from $s_2$.
This is, of course, the same for every $B_i$.

To recap, what makes interpreting $\P$ on a \blockdag so attractive:
sending blocks instead of messages in a deterministic $\P$ results in
a compression of messages---up to their omission. And not only do
these messages not have to be sent, they also do not have to be
signed. It suffices, that every server signs their blocks. Finally, a
single block sent is interpreted as messages for a very large
number of parallel protocol instances.

%%% Local Variables:
%%% mode: latex
%%% TeX-master: "paper"
%%% End:

\section{Related Work} \label{sec:related-work}

%\subsection{Block DAG protocols}

The last years have seen many proposals based on \blockdag paradigms
(see~\cite{2020_wang_et_al} for an SoK)---some with commercial
implementations. We focus on the proposals closest to our work:
\hashgraph~\cite{2016_baird}, \blockmania~\cite{2018_danezis_et_al},
\aleph~\cite{2019_gagol_et_al}, and \flare
\cite{2019_rowan_et_al}. Underlying all of these systems is the same
idea: first, build a common \blockdag, and then locally interpret the
blocks and graph structure as communication for some protocol: %
\hashgraph encodes a consensus protocol in \blockdag structure, %
\blockmania~\cite{2018_danezis_et_al} encodes a simplified version of
PBFT~\cite{1999_castro_et_ala},
\aleph~\cite{2019_gagol_et_al} employs atomic broadcast and consensus, and
\flare \cite{2019_rowan_et_al} builds on federated byzantine agreement
from \stellar~\cite{2015_mazieres} combined with \blockdags to
implement a federated voting protocol.
Naturally, the correctness arguments of these systems focus on their
system, \eg, the correctness proof in \coq of byzantine consensus in
\hashgraph~\cite{2018_crary}. In our work, we aim for a different
level of generality: we establish structure underlying protocols which
employ \blockdags, \ie, we show that a \blockdag implements a reliable
point-to-point channel~(Section~\ref{sec:interpret}). To that end, and
opposed to previous approaches, we treat the protocol $\P$ completely as a
black-box, \ie, our framework is parametric in the protocol~$\P$.

%\subsection{\peerreview}

The idea to leverage deterministic state machines to replay the
behavior of other servers goes back to
\peerreview~\cite{2007_haeberlen_et_al}, where servers exchange logs
of received messages for auditing to eventually detect and expose
faulty behavior.
% every message from correct node to other is eventually
% received (\cmp Lemma~\ref{lem:eventually-valid})
This idea was taken up by \blockdag approaches---but with the twist to
leverage determinism to \emph{not} send those messages that can be
determined. This allows compressing messages to the extent of only
indicating that a message has been sent as we do in
Section~\ref{sec:interpret}.
However, we believe nothing precludes our proposed framework to be
adapted to hold equivocating servers accountable, drawing \eg, on
recent work from \polygraph to detect byzantine behavior
\cite{2020_civit_et_al}.

%\subsection{Threshold Logical Clock Abstraction \cite{2019_ford} }

While our framework treats the interpreted protocol~$\P$ as a
black-box, the recently proposed threshold logical clock
abstraction~\cite{2019_ford} allows the higher-level protocol to
operate on an asynchronous network as if it were a synchronous network
by abstracting communication of groups. Similar to our framework, also
threshold clocks rely on causal relations between messages by including
a threshold number of messages for the next time step. This would
roughly correspond to including a threshold number of predecessor
blocks. In contrast, our framework, by only providing the abstraction
of a reliable point-to-point link to $\P$, pushes reasoning about
messages to $\P$.

%%% Local Variables:
%%% mode: latex
%%% TeX-master: "paper"
%%% End:

\section{Extensions, Limitations \& Conclusion} \label{sec:conclusion}

We have presented a generic formalization of a \blockdag and its properties, and
in particular results relating to the eventual delivery of all blocks from correct servers to other
correct servers. We then leverage this property to provide a concrete implementation of
a reliable point-to-point channel, which can be used to implement any deterministic protocol~$\P$ efficiently.
In particular we have efficient message compression, as those messages emitted by $\P$, which are the results of the deterministic
execution of $\P$ may be omitted. Moreover we are allowing for batching of the
execution of multiple parallel instances of $\P$ using the same \blockdag, and the de-coupling of
maintaining the joint \blockdag from its interpretation as instances of $\P$.

\paragraph{Extensions} First, throughout our work we assume $\P$ is deterministic.
The protocol may accept user requests, and emit deterministic messages based on these events and
other messages. However, it may not use any randomness in its logic.
It seems we can extend the proposed composition to
non-deterministic protocols $\P$---but some care needs to be applied around the security properties
assumed from randomness. In case randomness is merely at the discretion of a server running their
instance of the protocol we can apply techniques to de-randomize the protocol by relying on the server
including in their created block any coin flips used. In case randomness has to be unbiased, as is
the case for asynchronous Byzantine consensus protocols, a joint shared randomness protocol needs to be
embedded and used to de-randomize the protocol. Luckily, shared coin protocols that are secure under
BFT assumptions and in the synchronous network setting exist~\cite{2020_kokoriskogias_et_al} and our composition could be used
to embed them into the \blockdag. However we leave the details of a generic embedding for non-deterministic
protocols for future work.

Second, we have discussed the case of embedding asynchronous protocols into a \blockdag.
We could extend this result to BFT protocols in the partial synchronous
network setting~\cite{1988_dwork_et_al} by showing that the \blockdag interpretation not only creates
a reliable point-to-point channel but also that its delivery delay is bounded if the underlying
network is partially synchronous. We have a proof sketch to this effect, but a complete proof would
require to introduce machinery to reason about timing and, we believe, would not enhance the presentation of the core arguments behind our abstraction.

Third, our correctness conditions on the \blockdag seem to be much more strict than necessary.
For example, block validity requires a server to have processed all previous blocks. In practice
this results in blocks that must include at some position $k$ all predecessors of blocks to be
included after position $k$. This leads to inefficiencies: a server must include references to
all blocks by other parties into their own blocks, which represents an $O(n^2)$ overhead (admittedly
with a small constant, since a cryptographic hash is sufficient). Instead, block inclusion could
be more implicit: when a server $s$ includes a block $B'$ in its block $B$ all predecessors of
$B'$ could be implicitly included in the block $B$, transitively or up to a certain depth. This would
reduce the communication overhead even further. Since it is possible to take a \blockdag with this weaker
validity condition and unambiguously extract a \blockdag with the stronger validity condition
we assume, we foresee no issues for all our theorems to hold.
Furthermore, when interpreting a protocol currently a server only
indicates, when the server running the interpretation indicates in the
interpretation. This is to assure that the server running the
interpretation can trust the server in the interpretation, \ie
itself. Again, we believe that this can be weakened by leveraging
properties of the interpreted protocol.
However, we again leave a full exploration
of this space to future work. %

\paragraph{ Limitations} Some limitations of our composition require
much more foundational work to be overcome. And these limitations also
apply to the \blockdag based protocols which we attempt to formalize.
First, there are practical challenges when embedding protocols tolerating processes that can
crash and recover. At first glance safe protocols in the crash recovery setting seem like a great match for
the \blockdag approach: they do allow parties that recover to re-synchronize the \blockdag, and continue
execution, assuming that they persist enough information (usually in a local log) as part of $\P$. However
there are challenges: first, our \blockdag assumes that blocks issued have consecutive numbers. If the
higher-level protocols use these block sequence numbers as labels for state machines (as in \blockmania),
a recovering process may have to `fill-in' a large number of blocks before catching up with others.
An alternative is for block sequence numbers to not have to be consecutive, but merely increasing,
which would remove this issue.

However in all cases, unless there is a mechanism for the higher level protocol $\P$ to signal
that some information will never again be needed, the full \blockdag has to be stored by all
correct parties forever. This seems to be a limitation of both our abstraction of \blockdag but
also the traditional abstraction of reliable point-to-point channels and the protocols using them,
that seem to not require protocols to ever signal that a message is not needed any more (to stop
re-transmission attempt to crashed or Byzantine servers). Fixing this issue, and proving that protocols
can be embedded into a \blockdag, that can be operated and interpreted using a bounded amount of
memory to avoid exhaustion attacks is a challenging and worthy future avenue for work -- and is
likely to require a re-thinking of how we specify BFT protocols in general to ensure this property,
beyond their embedding into a \blockdag.

Finally, one of the advantages of using a \blockdag is the ability to separate the operation
and maintenance of the \blockdag from the later or off-line interpretation of instances of
protocol $\P$. However, this separation does not hold and extend to operations that change
the membership of the server set that maintain the \blockdag---often referred to as reconfiguration.
How to best support reconfiguration of servers in \blockdag protocols seems to be an open issue,
besides splitting protocol instances in pre-defined epochs.

%%% Local Variables:
%%% mode: latex
%%% TeX-master: "paper"
%%% End:

% ----------------------------------------------------------------
% you can include the acknowledgements in the source, but `anonymous' option will hide them
\begin{acks}
  This work has been partially supported the UK EPSRC Project
  EP/R006865/1, Interface Reasoning for Interacting Systems (IRIS).
\end{acks}

% ----------------------------------------------------------------
% use ACM-Reference-Format for the references
\bibliographystyle{ACM-Reference-Format}
\bibliography{references}

\appendix
\section{Appendix}

\subsection{Ad Section~\ref{sec:background}: Background}

\begin{definition} \label{def:crypto-hash} %
  Let $\hsh: A \to A'$ be a secure cryptographic hash function. We
  write $\hsh(x)$ for the hash of $x \in A$, and we write $\hsh(A)$
  for $A'$.
  By definition \cite[p.332]{1996_menezes_et_al}, for any~$\hsh$ it is
  computationally infeasible
  \begin{enumerate}
  \item \label{it:preimage-resistance} %
    to find any preimage $m$ such that $\hsh(m) = x$ when given any
    $x$ for which a corresponding input is not known
    (\emph{preimage-resistance}),
  \item \label{it:snd-preimage-resistance} %
    given $m$ to find a 2nd-preimage $m' \neq m$ such that
    $\hsh(m) = \hsh(m')$ (\emph{2nd-preimage resistance}), and
  \item \label{it:collision-resistance} %
    to find any two distinct inputs $m$, $m'$ such that
    $\hsh(m) =\hsh(m')$ (\emph{collision resistance}).  %
  \end{enumerate}
\end{definition}
\begin{proof}[Proof of Lemma~\ref{lem:ins:properties}~(\ref{lem:ins-idempotent})]
  By definition of $\G$ and $\ins$.
\end{proof}

\begin{proof}[Proof of Lemma~\ref{lem:ins:properties}~(\ref{lem:g:subseteq})]
  Let $\G' = \ins(\G, v, E)$.  By definition of $\ins$,
  $\vtcs_{\G} \subseteq \vtcs_{\G'}$. Assume $v \not\in \G$. As $E$
  contains only edges such that $(v_i, v)$ where $v \not\in \G$,
  $\edgs_{\G} = {\edgs_{\G'} \cap (\vtcs_{\G} \times \vtcs_{\G})}$
  holds.
\end{proof}

\begin{proof}[Proof of Lemma~\ref{lem:ins:properties}~(\ref{lem:g:acyclic})]
  By definition of $E$, $\ins(\G, v, E)$ only adds edges from vertices
  in $\G_1$ to $v$. As $v \not\in \G$, there is no edge $(v, v_j)$ in
  $\G$. By acyclicty of $\G$, $\ins(\G, v, E)$ is acyclic.
\end{proof}

\subsection{Ad Section~\ref{sec:block-dags}: Building a Block DAG}

In this section we give the proofs---and lemmas those proofs rely
on---which we omitted in Section~\ref{sec:block-dags}.
All proofs refer to Algorithm~\ref{alg:gossip}. For the execution we
assume, that the body of each handler is executed atomically and
sequentially within the handler.

% To illustrate the latter: a correct server will not send a block to
% every $s' \in \srvrs$ in line~\ref{lin:gossip:brdcst} before signing
% it in line~\ref{lin:gossip:sign}. For the former, \eg the execution
% of lines~\ref{lin:gossip:disseminate}--\ref{lin:gossip:next-b} will
% not be interrupted to add some $B$ to $\bB$ by
% line~\ref{lin:gossip:valid}.

\begin{proof}[Proof of Lemma~\ref{lem:mutual-excl-refs}.] %
  Let $x_1 = \refB(B_1)$ and $x_2 = \refB(B_2)$. By assumption,
  $x_1 \in B_2.\preds$. Assume towards a contradiction that
  $x_2 \in B_1.\preds$. Then, to compute $x_1$ we need to know
  $x_2 = \refB(B_2)$. But this contradicts preimage-resistance of
  $\refB$.
\end{proof}

\begin{lemma} \label{lem:ins-bdag-idempotent} %
  For a \blockdag $\G$ and a block $B \in \G$ holds $\G = \G.\ins(B)$,
  \ie $\ins$ is idempotent.
\end{lemma}
\begin{proof}
  By definition of $\ins$ on \blockdags $E$ is fixed to
  $\{ (B, B') \mid B \in B'.\preds \}$.
  Since $B \in \G$ also
  $\{ (B, B') \mid B \in B'.\preds \} \subseteq \edgs_\G$ by
  definition of \blockdag. Thus, $\G.\ins(B) = \G$ by
  Lemma~\ref{lem:ins:properties}~(\ref{lem:ins-idempotent}).
\end{proof}

\begin{lemma} \label{lem:insert-preserves-block-dag} %
  Let $\G$ be a \blockdag for a server $s$ and let $B'$ be a block
  such that $\valid(s, B')$ holds and for all $B \in B'.\preds$ holds
  $B \in \G$.
  Let $\G' = \G.\ins(B')$. Then $\G'$ is a \blockdag for $s$.
\end{lemma}
\begin{proof}
  To show $\G'$ is a \blockdag we need to show that $\G'$ adheres to
  Definition~\ref{def:block-dag}.
  % all-valid
  For condition~\ref{bdag:all-valid} we have to show that $s$
  considers all blocks in $\G'$ valid. By definition of $\ins$ holds
  $\vtcs_{\G'} = {\vtcs_{\G} \cup \{ B' \}}$. As $\G$ is a \blockdag
  for $s$, $\valid(s, B)$ holds for all $B \in \vtcs_\G$ and
  $\valid(s, B')$ follows from the assumption of the lemma.
  % correct-edges
  For condition~\ref{bdag:correct-edges} we have to show that for
  every backwards reference to $B$ from the block $B'$, the
  \blockdag~$\G'$ contains $B$ and an edge from $B$ to $B'$. The
  former---for all $B \in B'.\preds$ we have $B \in \G$---holds by
  assumption of the lemma. The latter---$(B, B') \in \edgs_{\G'}$ for
  $B \in B'.\preds$--- holds by definition of $\ins$. As $\G$ is a
  \blockdag, condition~\ref{bdag:correct-edges} holds for every block
  in $\G$.
  % acyclic
  It remains to show, that $\G'$ is acyclic. If $B' \in \G$ then by
  Lemma~\ref{lem:ins-bdag-idempotent}, $\G' = \G$ and $\G$ is
  acyclic. If $B' \not\in \G$ then $\G'$ is acyclic by
  Lemma~\ref{lem:ins:properties}~(\ref{lem:g:acyclic}).
\end{proof}

\begin{lemma} \label{lem:valid-bB} %
  For every correct server~$s$ executing $\gossip$ of
  Algorithm~\ref{alg:gossip}, whenever the execution reaches
  line~\ref{lin:gossip:insB} then $\valid(s, \bB)$ holds.
\end{lemma}
\begin{proof}
  We need to show, that once the execution reaches
  line~\ref{lin:gossip:insB}
  Definition~\ref{def:valid}~\ref{valid:sign-verified}--\ref{valid:preds}
  holds. %
  % valid-sign
  As $s$ is correct and signs $\bB$ in line~\ref{lin:gossip:sign}
  \ref{valid:sign-verified}~$\validsign(s, \bB.\signature)$ holds.
  % seqn
  We prove~\ref{valid:seqn} and \ref{valid:preds} by induction on the
  times~$n$ the execution reaches line~\ref{lin:gossip:insB}. For the
  base case, $\bB$ is \ref{valid:genesis} a genesis block with
  $\bB.\seqn = 0$ as initialized in line~\ref{lin:gossip:bb}. Moreover
  $\bB$ has no parent.  As $s$ is correct and only inserts $B'$ in
  $\bB.\preds$ in line~\ref{lin:gossip:ins-pred} whenever $s$
  considers $B'$ valid in line~\ref{lin:gossip:valid}, $s$ considers
  all $B' \in \bB.\preds$ valid.
  In the step case, $\bB_{n+1}$ is updated in
  line~\ref{lin:gossip:next-b}. We show that \ref{valid:parent}
  $\bB_{n+1}$ has exactly one parent~$\bB_{n}$. By
  line~\ref{lin:gossip:next-b}, $\bB_{n+1}.\nid = \bB_{n}.\nid$ and
  $\bB_{n+1}.\seqn = \bB_n.\seqn + 1$.
  As $\bB_n$ is inserted in $\bB_{n+1}.\preds$ in
  line~\ref{lin:gossip:next-b}, by definition
  $\bB_{n+1}.\parent = \bB_{n}$.
  % preds
  By induction hypothesis, $s$ considers $\bB_n$ valid, and again, as
  $s$ is correct and only inserts $B'$ in $\bB.\preds$ in
  line~\ref{lin:gossip:ins-pred} whenever $s$ considers $B'$ valid in
  line~\ref{lin:gossip:valid}, \ref{valid:preds} $s$ considers all
  $B' \in \bB.\preds$ valid.
\end{proof}

\begin{lemma} \label{lem:g-is-block-dag} %
  For every correct server~$s$ executing $\gossip$ of
  Algorithm~\ref{alg:gossip} $\G$ is a \blockdag.
\end{lemma}

\begin{proof} %
  We proof the lemma by induction on the times~$n$ the execution
  reaches line~\ref{lin:gossip:ins-b-g} or line~\ref{lin:gossip:insB}
  of Algorithm~\ref{alg:gossip}.
  % base case
  As $\G$ is initialized to the empty \blockdag in
  Algorithm~\ref{alg:shim} in line~\ref{lin:gossip:g}, $\G$ is a
  \blockdag for the base case $n = 0$.
  % step case
  In the step case, by induction hypothesis, $\G$ is a \blockdag. By
  Lemma~\ref{lem:insert-preserves-block-dag} $\G.\ins(B')$ is a
  \blockdag if
  \begin{enumerate*}[label=\emph{(\roman*)}]
  \item\label{it:is-valid} $\valid(s, B')$ holds, and
  \item\label{it:preds} for all $B \in B'.\preds$ holds $B' \in \G$.
  \end{enumerate*}
  The former~\ref{it:is-valid}, $\valid(s, B')$, holds either by
  line~\ref{lin:gossip:valid} or by Lemma~\ref{lem:valid-bB}.
  As $s$ inserts any block $B$ which $s$ has received and considers
  valid by lines~\ref{lin:gossip:valid}--\ref{lin:gossip:ins-pred},
  for the latter~\ref{it:preds} it suffices to show that $s$ considers
  all $B \in B'.\preds$ valid. As $s$ considers $B'$ valid, by
  Definition~\ref{def:valid}~\ref{it:preds}, $s$ considers all
  $B \in B'.\preds$ valid.
\end{proof}

\begin{proof}[Proof of Lemma~\ref{lem:eventually} (\ref{lem:eventually-receive})] %
  By assumption~$s$ considers $B$ valid, and hence by
  lines~\ref{lin:gossip:valid}--\ref{lin:gossip:ins-pred} adds a
  reference to $B$ to $\bB$. As $s$ is correct, $s$ eventually will
  $\disseminate()$, and then $s$ disseminates $\bB$ in
  line~\ref{lin:gossip:brdcst}. We refer to this disseminated $\bB$ as
  $B'$.  By Assumption~\ref{ass:reliable-links}, every correct server
  will eventually receive $B'$. Assume a correct server~$s'$, which
  has received~$B'$, but has not received $B$.
  As $s'$ has not received $B$, by
  Definition~\ref{def:valid}~\ref{valid:preds}, $s'$ does not consider
  $B'$ valid.  After time $\timer{B'}$ by
  lines~\ref{lin:gossip:missing-preds}--\ref{lin:gossip:send-fwd} $s'$
  will request $B$ from $s$ by sending $\FWD\ B$. Again by
  Assumption~\ref{ass:reliable-links}, after $s$ receives $\FWD\ B$
  from $s'$ by lines~\ref{lin:gossip:rcv-fwd}--\ref{lin:gossip:fwd},
  $s$ will send $B$ to $s'$, which will eventually arrive, and $s'$
  receives $B$.
\end{proof}

\begin{proof}[Proof of Lemma~\ref{lem:eventually} (\ref{lem:eventually-valid})] %
  We have to show, that $\valid(s', B)$ eventually holds for all
  correct servers~$s'$.
  For Definition~\ref{def:valid}~\ref{valid:sign-verified}, as $s$
  considers $B$ valid and $s$ is correct, $B$ has a valid
  signature. This can be checked by every $s'$.
  We show
  Definition~\ref{def:valid}~\ref{valid:seqn}~\ref{valid:genesis} and
  \ref{valid:preds} by induction on the sum of the length of the paths
  from genesis blocks to $B$. For the base case, $B$ does not have
  predecessors. As $s$ considers $B$ valid, then $B$ is a genesis
  block, and $s'$ will consider $B$ a genesis block, so
  Definition~\ref{def:valid}~\ref{valid:seqn}~\ref{valid:genesis} and
  \ref{valid:preds} hold.
  For the step case, let $B' \in B.\preds$.  By
  Lemma~\ref{lem:eventually}~(\ref{lem:eventually-receive}), every
  correct server~$s'$ will eventually receive~$B'$. By induction
  hypothesis, $s'$ will eventually consider $B'$ valid. The same
  reasoning holds for every $B' \in B.\preds$. It remains to show that
  $B$ has exactly one parent or is a genesis block. Again, this
  follows by $s$ considering $B$ valid. As $B.\parent \in B.\preds$
  $s'$ also considers $B.\parent$ valid.
\end{proof}

\begin{lemma} \label{lem:ref-at-most-once} %
  For every block~$B$ every correct server~$s$ executing $\gossip$ of
  Algorithm~\ref{alg:gossip} inserts $\refB(B)$ at most once in any
  block $B'$ with $B'.\nid = s$.
\end{lemma}
\begin{proof}
  By line~\ref{lin:gossip:rcv-b} of Algorithm~\ref{alg:gossip}, a correct
  server adds a block~$B$ to $\buffer$ only if $B \not\in \G$, and as
  $\buffer$ is a set, $B$ appears at most once in $\buffer$.
  Either $B$ remains in $\buffer$, or by
  lines~\ref{lin:gossip:valid}--\ref{lin:gossip:ins-pred}, for any
  block~$B'$ with $B'.\nid = s$, after $\refB(B)$ is inserted in
  $B'$, $B \in \G$ holds.
  Thus, for no future execution $B \not\in \G$ holds and therefore
  $B \not\in \buffer$. As $s$ is correct, it will not enter
  lines~\ref{lin:gossip:valid}--\ref{lin:gossip:ins-pred} again
  for~$B$.
\end{proof}

\begin{lemma} \label{lem:cupdag} %
  Let $s$ and $s'$ be correct servers with \blockdags $\G_s$ and
  $\G_{s'}$. Then their joint \blockdag
  $\G \geqslant \G_s \cup \G_{s'}$ is a \blockdag for $s$.
\end{lemma}

\begin{proof}
  Let $\bs = B_1, \ldots, B_{k-1}$ be blocks such that
  $B_i \in \G_{s'}$ but $B_i \not\in \G_s$ for $1 \leqslant i < k$. We
  show the statement by induction on $|\bs|$. As $\G_s$ is a \blockdag
  for $s$, the statement holds for the base case. For the step case we
  pick a $B_i \in \bs$ such that $B_i.\preds \cap \bs = \varnothing$.
  Such a $B_i$ exists, as in the worst case, $\G_s$ and $\G_{s'}$ are
  completely disjoint and $B_i$ is a genesis block in $\G_s$.
  It remains to show that $s$ considers $B_i$ valid and all
  $B_i.\preds$ are in $\G_s$. Then by
  Lemma~\ref{lem:insert-preserves-block-dag} $\G_s.\ins(B_i)$ is a
  \blockdag and by induction hypothesis the statement holds.
  For all $B' \in B_i.\preds$ holds $B' \in \G_s$ by definition of
  $\bs$. Moreover, as $\G_s$ is the \blockdag of $s$, $s$ considers
  every $B'$ valid. Then by \ref{valid:preds} of
  Definition~\ref{def:valid}, together with the fact that $s'$ is
  correct therefore \ref{valid:sign-verified} and \ref{valid:seqn}
  hold for $s$, $s$ considers $B_i$ valid.
\end{proof}

\begin{lemma} \label{lem:witness} %
  If $B_1 \in \G$ for the \blockdag $\G$ of a correct server~$s$, then
  eventually for a \blockdag $\G'$ of $s$ where $\G' \geqslant \G$
  holds $B_2 \in \G'$ and $B_2.\nid = s$ and $B_1 \nxt B_2$.
\end{lemma}

\begin{proof}
  For a correct server~$s$ it holds that $B_1 \in \G$ only after $s$
  inserted $B_1$
  either in line~\ref{lin:gossip:ins-b-g} or in
  line~\ref{lin:gossip:insB}.
  Then by either line~\ref{lin:gossip:ins-pred} or
  \ref{lin:gossip:next-b}, respectively, $B_1 \in \bB.\preds$ for
  $\bB.\nid = s$.
  As $s$ is correct $s$ will eventually call $\disseminate()$ and
  $s$ will reach line~\ref{lin:gossip:insB} for $\bB$ and insert $\bB$
  to $\G$ for some $\G' \geqslant \G$.
\end{proof}

\subsection{Ad Section~\ref{sec:interpret}: Interpreting a Protocol}

In this section we give the proofs---and lemmas those proofs rely
on---which we omitted in Section~\ref{sec:interpret}.
All proofs refer to Algorithm~\ref{alg:interpret}. For the execution
we assume, that the body of each handler is executed atomically and
sequentially within the handler.

\begin{lemma} \label{lem:uninterpreted} %
  For $B \in \G$ if $\interpreted[B] = \false$ then
  $B.\bfr[d, \ell] = \varnothing$ and $B.\pis[\ell] = \bot$ for
  $\ell \in \lbls$ and $d \in \{\In, \Out\}$.
\end{lemma}

\begin{proof}
  For every $B$, $\ell \in \lbls$, and $d \in \{\In, \Out\}$,
  initially we have $B.\bfr[d, \ell] = \varnothing$ and
  $B.\pis[\ell] = \bot$.  Assume towards a contradiction that
  $B.\bfr[d, \ell] \neq \varnothing$ or $B.\pis[\ell] \neq \bot$. As
  $B.\bfr[d, \ell]$ and $B.\pis[\ell]$ are only modified in
  lines~\ref{lin:interpret:copy-parent}--\ref{lin:interpreted-b} after
  $B$ is picked in line~\ref{lin:interpret:loop}, then by
  line~\ref{lin:interpreted-b} $\interpreted[B] = \true$ contradicting
  $\interpreted[B] = \false$.
\end{proof}

\begin{lemma} \label{lem:picked-eventually} %
  For a block $B \in \G$ and a correct server executing $\interpret(\G, \P)$
  in Algorithm~\ref{alg:interpret} every $B$ is eventually picked in
  line~\ref{lin:interpret:loop}.
\end{lemma}

\begin{proof}
  To pick $B$ in line~\ref{lin:interpret:loop}, $\eligible(B)$ has to
  hold. As $\G$ is finite and acyclic, every $B \in \G$ is
  $\eligible(B)$ eventually.
\end{proof}

\begin{lemma} \label{lem:unchanged} %
  For a block~$B \in \G$ and an $\ell \in \lbls$, if $\interpreted[B]$
  holds,
  \begin{enumerate*}
  \item \label{lem:unchanged:interpreted-bfr} %
    then $B.\bfr[d, \ell]$ will never be modified again for every
    $d \in \{ \In, \Out\}$.
  \item \label{lem:unchanged:interpreted-pis} %
    then $B.\pis[\ell]$ will never be modified again.
  \end{enumerate*}
\end{lemma}

\begin{proof}
  For part~\ref{lem:unchanged:interpreted-bfr}, assume that
  $B.\bfr[d, \ell]$ is modified. This can only happen in
  lines~\ref{lin:interpret:initbfr}, \ref{lin:interpret:in}, and
  \ref{lin:interpret:bfr} and only for $B$ picked in
  line~\ref{lin:interpret:loop}. But as $\interpreted[B]$, $B$ cannot
  be picked in line~\ref{lin:interpret:loop}, leading to a
  contradiction.
  For part~\ref{lem:unchanged:interpreted-pis} assume that
  $B.\pis[d, \ell]$ is modified. This can only happen in
  lines~\ref{lin:interpret:copy-parent} and \ref{lin:interpret:bfr},
  and only for $B$ picked in line~\ref{lin:interpret:loop}. But as
  $\interpreted[B]$, $B$ cannot be picked in
  line~\ref{lin:interpret:loop}, leading to a contradiction.
\end{proof}

\begin{lemma} \label{lem:out-label} %
  If $m \in B.\bfr[\Out, \ell]$ then there is a block $B'$ such that
  $(\ell, r) \in B'.\rs$ and $B' \nxtS B$.
\end{lemma}

\begin{proof}
  In Algorithm~\ref{alg:interpret}, $m \in B.\bfr[\Out, \ell]$ only after
  the execution reaches either
  \begin{enumerate*}
  \item line~\ref{lin:interpret:initbfr}, and then $B' = B$, or
  \item line~\ref{lin:interpret:bfr}, end then by
    line~\ref{lin:interpret:all-pis} exists a $B_j$ such that
    $(\ell_j, r) \in B_j.\rs$ for a label
    $\ell \in \{ \ell_j \mid {(\ell_j, r_j) \in B_j.\rs} \land {B_j
        \in \G} \land {B_j \nxtT B} \}$.
  \end{enumerate*}
\end{proof}

\begin{lemma} \label{lem:pis:nid} %
  For all $B.\pis[\ell] \neq \bot$ holds that $B.\pis[\ell]$ was
  started with $\P(\ell, B.\nid)$.
\end{lemma}

\begin{proof}
  Either
  \begin{enumerate*}[label=\emph{(\roman*)}]
  \item $B$ is a genesis block, and then by assumption started with
    $B.\nid$ and $\ell$, or
  \item $B$ has a parent and by line~\ref{lin:interpret:copy-parent},
    $\pis[\ell]$ is copied from $B.\parent$ and as
    $B.\parent.\nid = B.\nid$, $B.\pis[\ell]$ was initialized with
    $B.\nid$ and $\ell$ (Lemma~\ref{lem:pis:is-init}).
  \end{enumerate*}
\end{proof}

\begin{lemma} \label{lem:out-sender} %
  If $m \in B.\bfr[\Out, \ell]$ then $m.\sender = B.\nid$.
\end{lemma}

\begin{proof}
  By lines~\ref{lin:interpret:initbfr} and \ref{lin:interpret:bfr} of
  Algorithm~\ref{alg:interpret} $m \in B.\bfr[\Out, \ell]$ if either
  $m \in B.\pis[\ell].(B.\rs)$ or $m \in B.\pis[\ell].\receive(m')$
  for some $m'$ of no importance. Important is, that $B.\pis[\ell]$
  was initialized by $B.\nid$ by Lemma~\ref{lem:pis:nid}, and thus
  every out-going message~$m$ has $m.\sender = B.\nid$. It remains to
  show that every $B$ with $B.\nid = s$ was build by $s$, which
  follows by the signature~$B.\nid$.
\end{proof}

\begin{lemma} \label{lem:pis:is-init} %
  When the execution of $\interpret(\G, \P)$ reaches
  line~\ref{lin:interpret:all-pis} of Algorithm~\ref{alg:interpret}
  then for all
  $\ell_j \in \{ \ell_j \mid {(\ell_j, r) \in B_j.\rs} \land {B_j \in
    \G} \land {B_j \nxtS B} \}$ holds $B.\pis[\ell_j] \neq \bot$.
\end{lemma}

\begin{proof}
  We show the statement by induction on the length of the longest path
  from the genesis blocks to $B$.
  The base cases $n = 0$ holds by assumption, as $\pis[\ell]$ is
  started on every genesis block.
  For the step case, by induction hypothesis the statement holds for
  $B_i \in B.\preds$, and as $B.\parent \in B.\preds$ by
  line~\ref{lin:interpret:copy-parent} the statement holds.
\end{proof}

\begin{proof}[Proof of Lemma~\ref{lem:bfr}(\ref{lem:send-bfr})]
  By definition $s_1$ sends $m$ for some protocol instance~$\ell'$ if
  $s$ reaches in Algorithm~\ref{alg:interpret} either
  line~\ref{lin:interpret:initbfr} with $B.\rs$, or
  line~\ref{lin:interpret:bfr} with $B.\pis[\ell'].\receive(m)$ for
  some $B$ picked in line~\ref{lin:interpret:loop}. By
  Lemma~\ref{lem:pis:is-init} $B.\pis[\ell'] \neq \bot$ and
  $B.\pis[\ell'].\nid = s_1$ by assumption, by Lemma~\ref{lem:pis:nid}
  $B.\nid = s_1$.
  $B$ will be our witness for $B_1$.
  Now $m \in B.\bfr[\Out, \ell']$, by the assignment in either
  line~\ref{lin:interpret:initbfr} with $(\ell', r) \in B.\rs$ (by
  line~\ref{lin:interpret:l}), or in line~\ref{lin:interpret:bfr} with
  $(\ell', r) \in B_j.\rs$ for some $B_j \nxtT B$ (by
  line~\ref{lin:interpret:all-pis}).
  $B_j$ is our witness for $B' \neq B_1$.
  For the other direction, we have $B_1\in \G$ with $B_1.\nid = s_1$
  such that $m \in B_1.\bfr[\Out, \ell']$ for a $B' \in \G$ with
  $(\ell', r) \in B'.\rs$ and $B' \nxtS B_1$. By
  Lemma~\ref{lem:picked-eventually}, eventually $B_1$ is picked in
  Algorithm~\ref{alg:interpret} line~\ref{lin:interpret:loop}. By
  assumption, $m \in B_1.\bfr[\Out, \ell']$ through either
  \begin{enumerate*}[label=\emph{(\roman*)}]
  \item line~\ref{lin:interpret:initbfr}, or
  \item as $B' \nxtT B_1$ and thus
    $\ell' \in \{ \ell_j \mid {(\ell_j, r) \in B_j.\rs} \land {B_j \in
      \G} \land {B_j \nxtT B} \}$ from line~\ref{lin:interpret:bfr}.
  \end{enumerate*}
  Then, by definition, $s_1$ sends $m$ for protocol instance~$\ell'$.
\end{proof}

\begin{proof}[Proof of Lemma~\ref{lem:bfr}(\ref{lem:receive-bfr})]
  By Definition $s_2$ receives $m$ in line~\ref{lin:handle:receive} of
  Algorithm~\ref{alg:interpret} for protocol instance~$\ell'$
  for some $B$ picked in line~\ref{lin:interpret:loop} and
  $m \in B.\bfr[\In, \ell']$ by line~\ref{lin:handle:foreach}. By
  Lemma~\ref{lem:pis:is-init} $B.\pis[\ell'] \neq \bot$ and
  $B.\pis[\ell'].\nid = s_2$ by assumption, by
  Lemma~\ref{lem:pis:nid} $B.\nid = s_2$.
  $B$ is our witness for $B_2$.
  Now by line~\ref{lin:interpret:in} $m \in B.\bfr[\In, \ell']$ only
  if $m \in B_i.\bfr[\Out, \ell']$ for some $B_i$ with $B_i \nxt
  B$. $B_i$ is our witness for $B_1$.
  Finally, by line~\ref{lin:interpret:all-pis},
  $\ell' \in \{ \ell_j \mid {(\ell_j, r) \in B_j.\rs} \land {B_j \in \G} \land {B_j \nxtT B}
  \}$, and $B_j$ is our witness for $B'$.
  For the other direction we have $B_1, B_2 \in \G$ with
  $B_1 \nxt B_2$ and $B_2.\nid = s_2$ and
  $m \in B_2.\bfr[\In, \ell']$ for a $B' \in \G$ such that
  $(\ell', r) \in B'.\rs$ and $B' \nxtS B_1$.  By
  Lemma~\ref{lem:picked-eventually}, eventually $B_1$ is picked in
  Algorithm~\ref{alg:interpret} line~\ref{lin:interpret:loop} and by
  assumption eventually reaches line~\ref{lin:interpret:bfr} of
  Algorithm~\ref{alg:interpret}.
  As $m \in B_2.\bfr[\In, \ell']$ by definition, $s_2$ receives $m$
  for protocol instance~$\ell'$.
\end{proof}

\begin{lemma} \label{lem:send-transfer} %
  For a correct server $s$ executing $s.\interpret(\G, \P)$ if a
  server $s_1$ sends a message~$m$ for a protocol instance $\ell_j$,
  then $s_1$ sends $m$ for a correct server~$s'$ executing
  $s'.\interpret(\G', \P)$ for a \blockdag $\G' \geqslant \G$.
\end{lemma}

\begin{proof}
  Again, in the following proof, we write $\bfr'$ and $\pis'$ when
  executing $s'.\interpret(\G', \P)$ to distinguish from $\bfr$ and
  $\pis$ when executing $s.\interpret(\G, \P)$.
  As $s_1$ sends a message~$m$ for a protocol instance $\ell_j$, by
  Lemma~\ref{lem:bfr}(\ref{lem:send-bfr}) there is a $B_1 \in \G$ with
  $B_1.\nid = s_1$ such that $m \in B_1.\bfr[\Out, \ell_j]$ for a
  $B_j \in \G$ with $(\ell_j, r) \in B_j.\rs$ and $B_j \nxtN{n} B_1$
  for $n \geqslant 0$.
  By $\G' \geqslant \G$, $B_1 \in \G$, $B_j \in \G$, and the path
  $B_j \nxtN{n} B_1$ are in $\G'$.
  By Lemma~\ref{lem:joint-eq} $m \in B_1.\bfr'[\Out, \ell_j]$, and
  then by Lemma~\ref{lem:bfr}(\ref{lem:send-bfr}), $s_1$ sends s $m$ for a
  correct server~$s'$ executing $s'.\interpret(\G', \P)$.
\end{proof}

\begin{proof} [Proof of Lemma~\ref{lem:point-to-point}~(\ref{al:rd}) (Reliable delivery)]
  By assumption $s_1$ sends a message $m$ to a correct server~$s_2$
  for a correct server~$s$ executing $s.\interpret(\G, \P)$.
  By Lemma~\ref{th:joint-dag} $s'$ will eventually have some
  $\G_1 \geqslant \G$.
  Then by Lemma~\ref{lem:send-transfer}, $s_1$ sends $m$ in
  $s'.\interpret(\G_1, \P)$ for $\G_1 \geqslant \G$.
  Then by Lemma~\ref{lem:bfr}(\ref{lem:send-bfr}) there is a
  $B_1 \in \G_1$ with $B_1.\nid = s_1$ such that
  $m \in B_1.\bfr[\Out, \ell_j]$ for $B_j \in \G_1$ with
  $(\ell_j, r) \in B_j.\rs$ and $B_j \nxtS B_1$.
  With $B_1$ we found our first witness.
  By Lemma~\ref{lem:witness}, there is $\G_2 \geqslant \G_1$ such that
  $B_2 \in \G_2$ and $B_2.\nid = s_2$ and $B_1 \nxt B_2$.
  Then by Lemma~\ref{th:joint-dag} eventually $s'$ will have some
  $\G' \geqslant \G_2$.
  By $m \in B_1.\bfr[\Out, \ell_j]$, $B_1 \nxt B_2$ and
  $m.\receiver = s_2$ by assumption, by
  lines~\ref{lin:interpret:in}--\ref{lin:handle:foreach} of
  Algorithm~\ref{alg:interpret} we have $B.m \in \bfr[\In, \ell_j]$.
  Now we have found our second witness~$B_2$.
  By Lemma~\ref{lem:bfr}(\ref{lem:receive-bfr}), $s_2$ receives $m$ in
  $s'.\interpret(\G', \P)$
\end{proof}

\begin{proof} [Proof of Lemma~\ref{lem:point-to-point}~(\ref{al:nd}) (No duplication)]
  We assume towards a contradiction, that $s_2$ received $m$ more than
  once.
  Then by Lem\-ma~\ref{lem:bfr}(\ref{lem:receive-bfr}) there are some
  $B_1, B_2 \in \G$ with $B_1 \nxt B_2$, $B_2.\nid = s_2$ and
  $m \in B_2.\bfr[\In, \ell]$, \emph{and} $B_1' \nxt B_2'$,
  $B_2'.\nid = s_2$ and $m \in B_2'.\bfr[\In, \ell]$ for a
  $B_j \in \G$ such that $(\ell, r) \in B_j.\rs$ and $B_j \nxtS B_1$,
  but $B_2 \neq B_2'$.
  That $s_2$ received the exact same message~$m$ twice is only
  possible, if $B_1 = B_1'$.
  %
  % \mas{this depends on the exact
  % definition of equality of messages, but assuming we have a
  % timestamp---or, in our case we could use $\refB{B}$---the message
  % could have come message from the same block.}
  %
  That is, $s_2$ built $B_2' \neq B_2$ and inserted $B_1$ in both,
  which contradicts Lemma~\ref{lem:ref-at-most-once} as $s_2$ is
  correct.
\end{proof}

\begin{proof}[Proof of Lemma~\ref{lem:point-to-point}~(\ref{al:at})
  (Authenticity)] %
  By Lemma~\ref{lem:bfr}(\ref{lem:receive-bfr}) there are some
  $B_1, B_2 \in \G$ with $B_1 \nxt B_2$ and $B_2.\nid = s_2$ and
  $m \in B_2.\bfr[\In, \ell]$ for a $B \in \G$ such that
  $(\ell, r) \in B_j.\rs$ and $B_j \nxtS B_1$.
  Then by line~\ref{lin:interpret:in} of Algorithm~\ref{alg:interpret}
  exists an $B_i \in B_2.\preds$ such that $m \in B_i.\bfr[\Out, \ell]$.
  As $m \in B_i.\bfr[\Out, \ell]$ by Lemma~\ref{lem:out-sender}
  $B_i.\nid = m.\sender$ and as $m.\sender = s_1$, $B_i.\nid = s_1$.
  $B_i$ will be our witness for $B_1$.
  As $m \in B_i.\bfr[\Out, \ell]$ by Lemma~\ref{lem:out-label} there
  is a $B'$ such that $(\ell, r) \in B'.\rs$ and $B' \nxtS B_i$. $B'$
  is our witness for $B_j$.
  Hence there is a $B_1 \in \G$ with $B_1.\nid = s_1$ such that
  $m \in B_1.\bfr[\Out, \ell]$ for a $B_1 \in \G$ with
  $(\ell, r) \in B_j.\rs$ and $B_j \nxtS B_1$ and by
  Lemma~\ref{lem:bfr}(\ref{lem:send-bfr}) $s_1$ $m$ was sent by $s_1$.
\end{proof}

\subsection{Ad Section~\ref{sec:interfacing}: Using the Framework}

\begin{algorithm}[t]
  \setcounter{AlgoLine}{0}
  \SubAlgo{\Module $\brdcst(s \in \srvrs)$}
  {
    $\echoed, \readied, \delivered \assign \false$ \\
    \smallskip %
    \SubAlgo{ $\broadcast(v \in \vals)$ \And $\authenticate(v)$ \label{lin:brdcst:auth} %
    } %
    { %
      $\echoed \assign \true$ \\
      \Send \To $\ECHO\ v$ \To \Every $s' \in \srvrs$}
    \smallskip %
    \SubAlgo{ %
      \Upon \Received $\ECHO\ v$ \And \Not $\echoed$ \label{lin:brdcst:r-echo} %
    } %
    { %
      $\echoed \assign \true$ \\
      \Send $\ECHO\ v$ \To \Every $s' \in
      \srvrs$ \label{lin:brdcst:s-echo}
    }
    \smallskip %
    \SubAlgo{ %
      \Upon \Received $\ECHO\ v$ \From $2f + 1$ different
      $s' \in \srvrs$ \And \Not $\readied$%
    } %
    { %
      $\readied \assign \true$ \\
      \Send $\READY\ v$ \To \Every $s' \in \srvrs$ %
    } %
    \smallskip %
    \SubAlgo{ %
      \Upon \Received $\READY\ v$ \From $f + 1$ different
      $s' \in \srvrs$ \And \Not $\readied$%
    } %
    { %
      $\readied \assign \true$ \\
      \Send $\READY\ r$ \To \Every $s' \in \srvrs$ %
    }
    \smallskip %
    \SubAlgo{ %
      \Upon \Received \From $\READY\ v$ \From $2f + 1$ different
      $s' \in \srvrs$ \And \Not $\delivered$%
    } %
    { %
      $\delivered \assign \true$ \\
      $\deliver(v)$
    }
  }
   \caption{Authenticated double-echo broadcast after
    \cite{2011_cachin_et_al}.}
  \label{alg:brdcst}
\end{algorithm}

In this section we give the proofs which we omitted in
Section~\ref{sec:interfacing}.
All proofs refer to Algorithm~\ref{alg:shim}. For the execution we
assume, that the body of each handler is executed atomically.
We further give an implementation of authenticated double-echo
broadcast in Algorithm~\ref{alg:brdcst}.

\begin{lemma} \label{lem:r-in-p} %
  For a correct server $s$ executing $\shim(\P)$, if some
  $\request(r, \ell)$ is requested from $s$, then $r$ is requested in
  $\P$.
\end{lemma}

\begin{proof}
  By executing $\shim(\P)$, a correct server~$s$ inserts $(\ell, r)$
  in $\requestBuffer$ in
  line~\ref{lin:shim:auth-r}--\ref{lin:shim:ins-r-buffer} of
  Algorithm~\ref{alg:shim}. Then executing
  $\gossip(s, \G, \requestBuffer)$, $s$ will eventually disseminate a
  block $B$ with $B.\nid = s$ and $(\ell, r) \in B.\rs$ in
  line~\ref{lin:gossip:inject} of Algorithm~\ref{alg:gossip} and
  $B \in \G$ after triggering~$\disseminate$ in
  lines~\ref{lin:repeatedly}--\ref{lin:disseminate} of
  Algorithm~\ref{alg:shim}. Now, executing $\interpret(\G, \P)$, $s$
  for $B \in \G$ will call $B.\pis[\ell].\rs$ in
  line~\ref{lin:interpret:initbfr} in Algorithm~\ref{alg:interpret}.
\end{proof}

\begin{lemma} \label{lem:i-in-shim-p} %
  For a correct server $s$ executing $\shim(\P)$, if $\P$
  indicates~$i \in \inds_\P$ for $s$, then $\shim(\P)$ triggers
  $\indicate(\ell, i)$.
\end{lemma}

\begin{proof}
  By assumption a correct $s$ indicates~$i$ for $\ell$ and hence
  indicates in $\interpret(\G, \P)$
  lines~\ref{lin:interpret:indicate}--\ref{lin:interpret:indicate-up}
  of Algorithm~\ref{alg:interpret}. Then, by executing $\shim(\P)$, as
  $s = s'$ $\indicate(\ell, i \in \inds_\P)$ is triggered in
  lines~\ref{lin:shim:indicate}--\ref{lin:shim:indicate-up} of
  Algorithm~\ref{alg:shim}.
\end{proof}

%%% Local Variables:
%%% mode: latex
%%% TeX-master: "paper"
%%% End:

\end{document}